\documentclass[12pt]{amsart}
\usepackage{graphicx}
\usepackage{graphicx}
\usepackage{epsfig}
\usepackage{pstricks}

\newtheorem{thm}{Theorem}[section]

\newtheorem{defn}[thm]{Definition}
\newtheorem{lemma}[thm]{Lemma}

\newtheorem{cor}[thm]{Corollary}
\newtheorem{remark}[thm]{Remark}
\newtheorem{example}[thm]{Example}

\usepackage{amsmath}
\usepackage{amsxtra}
\usepackage{amscd}
\usepackage{amsthm}
\usepackage{amsfonts}
\usepackage{amssymb}
\usepackage{eucal}
\newcommand{\bmb}{\left( \begin{array}{rr}}
\newcommand{\enm}{\end{array}\right)}

\newcommand{\cT}{\mathcal T}

\newcommand{\g}{{\mathfrak{g}}}

\newcommand{\cA}{\mathcal A}

\newcommand{\Z}{{\mathbb Z}}

\newcommand{\bm}{{\mathbf m}}
\newcommand{\bj}{{\mathbf j}}

\newcommand{\bx}{{\mathbf x}}
\newcommand{\by}{{\mathbf y}}

\newcommand{\bX}{{\mathbf X}}

\numberwithin{equation}{section}

\begin{document}
\title{The solution of the quantum $A_1$ $T$-system for arbitrary boundary}
\author{Philippe Di Francesco} 
\address{
Institut de Physique Th\'eorique du Commissariat \`a l'Energie Atomique, 
Unit\'e de Recherche associ\'ee du CNRS,
CEA Saclay/IPhT/Bat 774, F-91191 Gif sur Yvette Cedex, 
FRANCE. e-mail: philippe.di-francesco@cea.fr}

\author{Rinat Kedem}
\address{Department of Mathematics, University of Illinois MC-382, Urbana, IL 61821, U.S.A. e-mail: rinat@illinois.edu}
\date{\today}
\begin{abstract}
We solve the quantum version of the $A_1$ $T$-system by use
of quantum networks. The system is interpreted as a particular set of
mutations of a suitable (infinite-rank) quantum cluster algebra, and 
Laurent positivity follows from our solution.
As an application we re-derive the corresponding quantum network
solution to the quantum $A_1$ $Q$-system
and generalize it to the fully non-commutative case. We give the
relation between the quantum $T$-system and the quantum
lattice Liouville equation, which is the quantized $Y$-system.
\end{abstract}

\maketitle

\section{Introduction}\label{introduction}

The $T$-systems \cite{KR,KNS} satisfied by the transfer matrices of
the generalized Heisenberg model or the $q$-characters of quantum
affine algebras \cite{FR} can be considered as discrete
dynamical systems with special initial conditions. More generally, the
equations of these systems can be shown \cite{DFK08} to be mutations
in an infinite-rank cluster algebra \cite{FZI}. As such, their
solutions under general boundary conditions \cite{DKTsys,DF} are expected
to satisfy special properties such as the Laurent property and positivity.

Among discrete dynamical systems, the cluster algebras of Fomin and
Zelevinsky \cite{FZI} hold a special place. These describe the
evolution of data vectors (clusters) attached to the nodes of an
infinite regular tree via mutations along the edges.  Mutations are
defined in such a way that the following Laurent property is
guaranteed: any cluster data may be expressed as a Laurent polynomial
of the cluster variables at any node of the tree.  It was conjectured in \cite{FZI}
and proved in several particular cases (in particular in the so-called 
acyclic cases \cite{POSIT,FZIV,FRISES},
or that of clusters arising from surfaces \cite{MSW}) 
that these polynomials always
have non-negative integer coefficients (Laurent positivity), a
property that still awaits a good general combinatorial interpretation.

Cluster algebras turn out to be quite universal, and have found
applications in various fields, such as the study of non-linear
recursions, the geometry of Teichm\"uller space, quiver representations,
wall crossing formulas etc.

The relation between the recursion satisfied by the ($q$-) characters of
KR-modules of quantum affine algebras on the one hand, and cluster
algebras on the other, was found in \cite{Ke07,DFK08}. Such systems
are known as $Q$-sytems or $T$-systems when they are supplemented 
by special
boundary conditions. It is known that such equations can be
interpreted as discrete integrable systems: In the case of an $A_r$
type algebra, the $T$-system was identified as the discrete Hirota
equation \cite{KLWZ}. It is also known as the tetrahedron equation in
combinatorics, and arises in the context
of the Littlewood-Richardson coefficients for tensor products of
irreducible representations of $A_r$ \cite{KT} and domino tilings of the
Aztec diamond \cite{SPY}.

Solutions to the $Q$ and $T$-systems have been constructed by various
authors \cite{KLWZ}.  Recently, a transfer matrix solution was given
for the $A_1$ $T$-system in the case of arbitrary boundary
conditions. The latter is also known in the combinatorics literature as frieze
equation \cite{FRISES}. This solution was generalized to the case of
$A_r$ in \cite{DF}. It amounts to representing general solutions of
the system as partition functions for paths on a positively
weighted graph or network. The graph is determined solely by the initial
conditions.  

The connection to cluster algebras is as follows \cite{Ke07,DFK08}:
One shows that the admissible initial data for the $T$-systems form a
subset of the clusters a cluster algebra, and that the mutations in
this algebra are local transformations which are the $T$-system
equations. Thus the expression of the solutions as partition functions
for positively weighted paths implies the Laurent positivity for these
particular clusters.

An important question is how to quantize such evolution equations
\cite{FV,FKV}. The quantization in the case of cluster algebras
generally was given by Berenstein and Zelevinsky \cite{BZ}.  Quantum
cluster algebras are non-commutative algebras where the cluster
variables obey special commutation relations depending on a
deformation parameter $q$.  Mutations are defined in such a way that
the Laurent property is preserved, and a positivity conjecture is also
expected to hold: That is, any cluster variable may be expressed
as a Laurent polynomial of the variables in any other cluster seed, with
coefficients in $\Z_+[q,q^{-1}]$. Quantum cluster algebras were used
in \cite{DFK10} to define quantum $Q$-systems for $A_r$.

In the present paper, we focus on the $T$ system for $A_1$, and
construct its quantum version via the cluster algebra connection. 
We gather a few definitions in Section \ref{definitions} and construct
the quantum $A_1$ $T$-system in Section \ref{section3}.
We then express the general solution in Section \ref{section5}
by use of a non-commutative transfer
matrix, quantizing the solution of \cite{DF}. The main result
of the paper is Theorem \ref{solq},
which implies an interpretation of
the solution as a partition function for ``quantum paths" with
step weights which are non-commutative Laurent
monomials in the initial data, thereby proving Laurent positivity for the relevant
clusters.

The $2$-periodic solutions of the $T$-systems satisfy $Q$-system
equations, and this generalizes to the quantum case. The $A_1$
$Q$-system has a fully non-commutative generalization introduced by
Kontsevich in the framework of wall-crossing phenomena in
non-commutative Donaldson-Thomas invariant theory. The solution was
given in \cite{DFK09b} for this system using the method of
\cite{DFK3}.  We revisit this system in Section \ref{section6} and
formulate a fully non-commutative (as opposed to $q$-commutative)
version of the network transfer matrices used in Section 4 to solve the
$T$-system. This gives an alternative solution for the non-commutative
$A_1$ $Q$-system.

Finally in Section \ref{conclusion} we give the relation between the
quantum $A_1$ $T$-system and the discrete quantum Liouville equation
of Faddeev et al \cite{FV,FKV}. This equation can also be viewed as a
non-commutative $Y$-system.

\medskip \noindent{\bf Acknowledgments.} We thank L. Faddev for
illuminating remarks, and the organizers of the MSRI semester program
on ``Random Matrix Theory, Interacting Particle Systems and Integrable
Systems" where this work was completed.  PDF received partial support
from the ANR Grant GranMa. The work of RK is supported by NSF grant
DMS-0802511.

\section{Definitions}\label{definitions}

\subsection{Cluster algebras and quantum cluster algebras}

We use a simplified version of the definition of Fomin and Zelevinsky
\cite{FZI} of cluster algebras of geometric type with trivial
coefficients.

\subsubsection{Cluster algebras of geometric type}
A cluster algebra is the commutative ring generated by the union of
commutative variables called cluster variables. The generators are
related by rational transformations called
{\em mutations} determined by an {\em exchange matrix}, which governs the
discrete dynamics of the system. 

For the purposes of his paper, it is sufficient to consider a cluster
algebra of rank $n$, with a seed cluster consisting of $n$ cluster
variables $\bx=(x_1,...,x_n)$ and an $n\times n$ skew-symmetric
exchange matrix $B$. (We will also have occasion to consider cluster
algebras of infinite rank ($n\to\infty$). It will be clear from our
solution that when such algebras occur they are well-defined as a
completion of the finite rank case.)

Clusters are pairs $(\bx(t),B(t))$ where $t$ is a label of a node of a
complete $n$-tree. Each node is associated with a cluster. The edges
of the tree are labeled in such a way that each node is connected to
exactly one edge with label $k$ where $k\in [1,...,n]$.

The clusters at nodes $t$ and $t'$ connected by an edge labeled
$k$ are related to each other by a {\em mutation}, which acts as a rational
transformation on the component $x_k(t)$:
$\bx(t')=\mu_k (\bx(t))$ where
\begin{equation}\label{mutation}
x_j(t')= \left\{ \begin{array}{ll} x_j(t), & k\neq j,\\
(x_k(t))^{-1}\left(\displaystyle{\prod_{B_{j,k}>0} x_j(t)^{B_{j,k}} + \prod_{B_{j,k}<0} x_j(t)^{-B_{j,k}}}\right),&k=j.
 \end{array} \right.
\end{equation}
A mutation also acts on the exchange matrix $B(t)$, such that
$B(t')=\mu_k(B(t))$ if $t$ and $t'$ are connected by an edge labeled
$k$, and
\begin{equation}\label{Bmutation}
B_{i,j}(t')=\left\{ \begin{array}{ll} -B_{i,j}(t) & \hbox{if $i=k$ or $j=k$, }\\
B_{i,j}(t) + {\rm sign}(B_{i,k}(t))[B_{i,k}(t)B_{k,j}(t)]_+ & \hbox{otherwise}.
\end{array}\right.
\end{equation}
with the notation $[x]_+={\rm Max}(x,0)$.
To define a cluster algebra, it is sufficient to give the seed
$(\bx,B)$ at one single node. This then determines the cluster
variables at all other nodes via iterated mutations.

\subsubsection{Quantum cluster algebras}
It is interesting to consider whether there exist non-commutative
generalizations of cluster algebras, which maintain some of the
properties of cluster algebras. In particular, the Laurent property
\cite{FZLaurent} and the (conjectured in general) positivity of a
cluster variable at any node as an expression in terms of the cluster
variables at any other node. We considered some possible candidates in
\cite{DFK10} motivated by our consideration of the integrable
subcluster algebras described by $Q$-systems and $T$-systems, as well
as the Kontsevich wall-crossing formula. We also considered the
specialization of these non-commutative systems to the simplest type of
non-commutativity, the $q$-deformation. Such systems were first considered by
Berenstein and Zelevinsky in Ref.\cite{BZ}, where they defined ``quantum
cluster algebras''. We give here a simplified version of their
definition which is sufficient for this paper.

A quantum cluster algebra is the skew field of rational functions
generated by the non-commutative cluster variables
$\{\bX(t)=(X_1(t),...,X_n(t))\}$ where $t$ are the labels of the
complete $n$-tree as above. At a node $t$ we have the cluster
$(\bX,B):=(\bX(t),B(t))$ where the exchange matrix is the same as for the usual
cluster algebra. The cluster variables at this node $q$-commute:
\begin{equation}\label{qcommute}
X_i X_j = q^{\lambda_{ij}} X_j X_i.
\end{equation}
Here $\lambda_{ij}$ are the entries of an $n\times n$ skew-symmetric
matrix $\Lambda$. Up to a scalar multiple, we can take $\Lambda$ to be
the inverse of the exchange matrix $B$. According to the definitions
of \cite{FZI} such a matrix $\Lambda$ is ``compatible'' with the
exchange matrix $B$. 

Equivalently, we can define $X_j=e^{a_j}$, where $a_j$ are also
non-commuting variables and the exponential is taken formally. Then the
commutation relations above correspond to $[a_i,a_j]=h \lambda_{i,j}$
where $q=e^h$. 

Clusters at tree nodes $t$ and $t'$ connected by an edge labeled $k$
are related by a mutation. Let $\bX^{\boldsymbol{\alpha}}={\rm
  exp}(\sum_j \alpha_j a_j)$. Then we define $\bX(t') = \mu_k
(\bX(t))$ to be
\begin{equation}\label{qmutation}
X_j(t') = \left\{ \begin{array}{ll}
X_j(t) & \hbox{if $j\neq k$}; \\
\bX^{-e_k + \sum_i [B_{ik}]_+ e_i} + \bX^{-e_k + \sum_i [-B_{ik}]_+e_i}& j=k.\end{array}\right.
\end{equation}
The exchange matrix $B(t')=\mu_k(B(t))$ is the same as in the
commutative case.

\subsection{The $A_1$ $T$-system}

The $T$-systems appear in the solution of exactly solvable
models in statistical mechanics, in the Bethe ansatz of generalized
Heisenberg quantum spin chains based on representations of Yangians
of each simple Lie algebra \cite{KR, KNS}. The transfer
matrices of the model satisfy a recursion relation in the highest
$\g$-weight of the $Y(\g)$-modules corresponding to the auxiliary
space. These relations are called $T$-systems.  In the context of
representation theory, these relations are the equations satisfied by
the $q$-characters \cite{FR} of Kirillov-Reshetikhin modules of
the Yangians, or the associated quantum affine algebra. 

\subsubsection{The $T$-system associated to $A_1$}

These systems provide examples of discrete integrable
systems which are part of a suitable cluster algebra structure
\cite{DFK08}. However, in the representation-theoretical context, a
special initial condition is placed on the variables (corresponding
to the fact that the $q$-character of the trivial representation is
1). Here, we dispense with this special value. Moreover, we
renormalize the variables so that the solutions are positive Laurent
polynomials of the initial data for any initial data. This corresponds
to normalizing the cluster variables of the cluster algebras so that
all coefficients are trivial. 

Thus, with slight abuse of notation, we call the following system the
$A_1$ $T$-system:
\begin{equation}\label{clatsys}
T_{i,j+1}T_{i,j-1}=T_{i+1,j}T_{i-1,j}+1\qquad (i,j \in \Z).
\end{equation}
Here, we consider the set $\{T_{i,j}\vert i,j\in\Z\}$ as commutative
variables. Solutions of the equation are given as functions
of a choice of initial variables. 

\begin{remark}
Upon a simple change of coordinates, this system is also known as
``frieze" equation in combinatorics \cite{Cox,FRISES}.
\end{remark}

\subsubsection{Initial conditions}

Equations \eqref{clatsys} split into two independent
sets of recursion relations, since the parity of $i+j$ is
preserved by the equations. Without loss of generality, let us
restrict to the relations for $\{T_{i,j}\vert i,j\in\Z,\  i+j\equiv0\! \mod 2\}$.
\begin{defn} An admissible initial data set for the $T$-system is a set
\begin{equation}\label{initial}
\bx_{\mathbf j}:=\{T_{i,j_i}\vert  i\in \Z,\ 
i+j_i\equiv 0\!\! \mod 2,\  |j_i-j_{i+1}|=1\}.
\end{equation}
\end{defn}
The solutions of Equation \eqref{clatsys} are determined by iterations
of the evolution equations \eqref{clatsys} starting from
any admissible initial data set.

\begin{defn}
The fundamental initial data (the ``staircase'') is the set
\begin{equation}\label{fundamental}
\bx_0 := \{T_{i,i \,{\rm mod}\, 2}\vert i\in \Z\}
\end{equation}
\end{defn}

\begin{defn}
The boundary corresponding to the initial condition $\bx_{\mathbf j}$ is the
set of points in the lattice $\{(i,j_i)\vert i\in \Z,\  i+j_i\equiv 0\!\! \mod 2\}$.
\end{defn}

A solution of the $T$-system is an expression for $T_{i,j}$ in terms
of $\bx_{\mathbf j}$ for each $(i,j)$.  A general solution of the
$A_1$ $T$-system for arbitrary boundary was given in \cite{FRISES} and
generalized to the case of the $A_r$ algebra, $r\geq 2$, in \cite{DF}.
The solution is given in terms of a matrix representation and
is interpreted as partition functions of networks.

In the present paper, we will introduce a quantum version of the $A_1$
$T$-system and derive its solutions in terms of quantum networks.


\subsubsection{The cluster algebra for the $A_1$ $T$-system}

The formulation of the $T$-systems as sub-cluster algebras was given in
\cite{DFK08}. In the case of $A_1$ the cluster algebra is given as follows.

\begin{defn}
Let $A$ be the cluster algebra of infinite rank generated by the
fundamental seed $(\bx_0,B_0)$, where $\bx_0$ is given by
\eqref{fundamental}, and the exchange matrix $B_0$
has entries
\begin{equation}\label{bzero}
(B_0)_{i,i'}= (-1)^i (\delta_{i',i+1}+\delta_{i',i-1})
\end{equation}
where the indices refer to the first index of the $T_{i,j}$'s.
\end{defn}

Each equation in the $T$-system corresponds to a mutation in the cluster
algebra $A$ (but not vice versa). 
All solutions of the $T$-system are contained in a subset
of the clusters corresponding to $x_{\mathbf j}$ defined in
\eqref{initial}. 
They are obtained from $(\bx_0,B_0)$ via iterated cluster
mutations of the form $\mu_{a}^{\pm}: \bx_\bj\to
\bx_{\bj'}$:
$$ \mu_a^+ :{\epsfxsize 3.cm \epsfbox{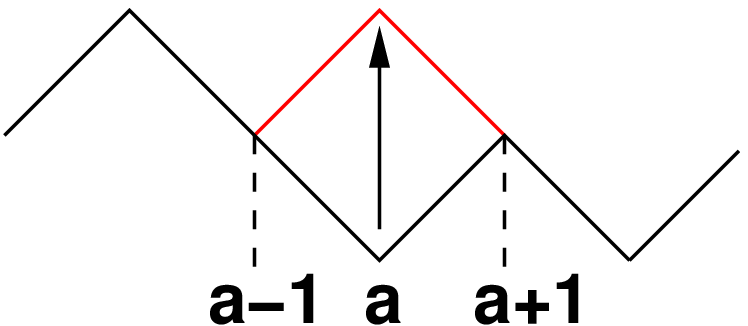}}\qquad\qquad 
\mu_a^{-} :{\epsfxsize 3.cm \epsfbox{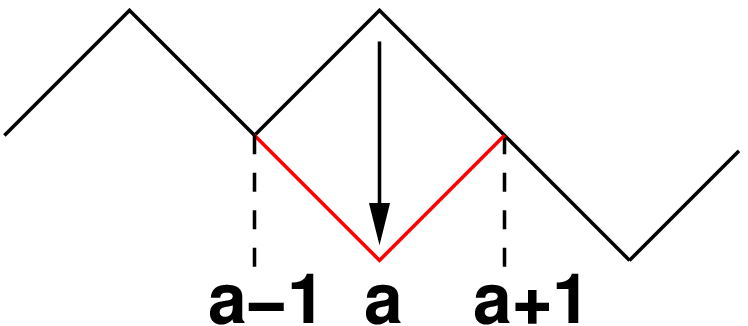}}$$
with $j'_b=j_b\pm 2\delta_{b,a}$, where $\mu_{a}^{\pm}$ leaves
all cluster variables unchanged apart from:
\begin{eqnarray*}
T_{a,j_a\pm2}:=\mu_a^{\pm}(T_{a,j_a})=(T_{a+1,j_a\pm1}T_{a-1,j_a\pm1}+1)/T_{a,j_a}
\end{eqnarray*}
in the case where all three terms on the right hand side are cluster
variables in $\bx_{\bj}$.

\section{The quantum $A_1$ $T$-system}\label{section3}


\subsection{Commutation relations in the initial seed}

In this note, we consider the non-commuting ``quantum" version of the
$A_1$ $T$-system.  Recall that for any cluster algebra of finite rank
$n$, a quantum cluster algebra is obtained by producing a ``compatible
pair" $(B_0,\Lambda)$ of skew-symmetric $n\times n$ integer matrices,
with $B_0\Lambda= d$, where $d$ is a diagonal matrix with positive
integer entries.  In turn, $\Lambda$ encodes the $q$-commutation
relations between the cluster variables of the initial cluster
$\bx_0=(x_i)_{i\in [1,n]}$ via $x_i x_j =q^{\Lambda_{i,j}}x_j x_i$.


\begin{figure}
\centering
\includegraphics[width=10.cm]{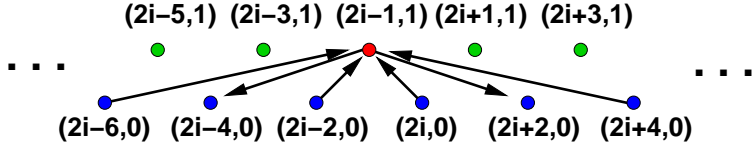}
\caption{\small The commutations between $T_{2i-1,1}$ and the other variables of the 
fundamental cluster $\bx_0$. Vertices $a$ and $b=(2i-1,1)$
are connected by an arrow $a\to b$ iff $T_aT_b =q T_bT_a$ and $b\to a$ iff $T_aT_b =q^{-1} T_bT_a$,
while $T_aT_b=T_bT_a$ otherwise.}
\label{fig:qcomz}
\end{figure}

However, the $A_1$ $T$-system comes from an infinite rank cluster
algebra with $\bx_0=(x_i)_{i\in \Z}$. We adapt the above condition,
based on the quantization of the $A_1$ $Q$-system \cite{DFK10,BZ}, which is a
specialization of the $T$-system.   

\begin{lemma}
Let $\Lambda$ be an infinite, skew-symmetric matrix such that
$$(B_0\Lambda)_{i,j}=(-1)^i(\Lambda_{i+1,j}+\Lambda_{i-1,j})=2\delta_{i,j},
\quad i,j\in\Z, 
$$
and such that $\Lambda_{i+m,i}=\Lambda_{i-m,i},\ (m>0)$.
Then
\begin{equation}\label{initlambda}
\Lambda_{i,j}= {1-(-1)^{i+j}\over 2} \times \left\{ \begin{matrix} (-1)^{i+j-1\over 2} & {\rm if} \ i\geq j\\
(-1)^{i+j+1\over 2} & {\rm if} \ i<j \end{matrix} \right.
\end{equation}
\end{lemma}
The second condition on $\Lambda$ is a choice, a reflection symmetry
imposed on the matrix $\Lambda$ which determines the matrix entries
\eqref{initlambda} completely.

The matrix $\Lambda$ encodes the commutation relations among the
elements of the 
fundamental cluster variable $\bx_0$
\begin{eqnarray}\label{qcomzero}
T_{2i-2k,0}T_{2i+1,1}&=&q^{(-1)^k}\, T_{2i+1,1}T_{2i-2k,0},\nonumber\\
T_{2i+2k,0}T_{2i-1,1}&=&q^{(-1)^k}\, T_{2i-1,1}T_{2i+2k,0}, \quad i\in\Z, k\geq0,
\end{eqnarray}
These commutation relations are
depicted graphically in Figure \ref{fig:qcomz}.

\subsection{The quantum $A_1$ $T$-system}

\begin{figure}
\centering
\includegraphics[width=10.cm]{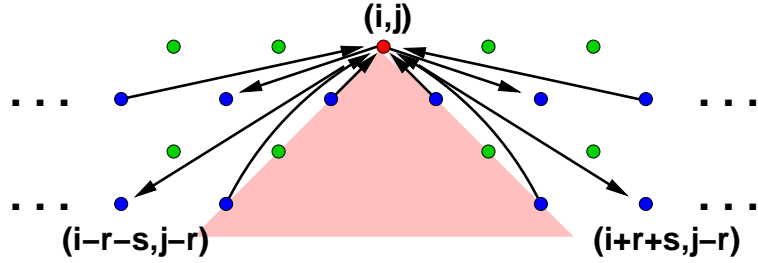}
\caption{\small The q-commutations between $T_{i,j}$ and
  $T_{i+r+s,j-r}$ or $T_{i-r-s,j-r}$ for $r,s\geq 0$ are illustrated
  as follows: vertices $a$ and $b$ are connected by an arrow $a\to b$
  iff $T_aT_b =q T_bT_a$. Note that only vertices $a,b$ with heights
  $j$ of opposite parity give rise to non-trivial commutations. The interior of the
  shaded cone below $(i,j)$ corresponds to the values $(k,\ell)$ such that
  $T_{k,\ell}$ cannot
  belong to the same cluster as $T_{i,j}$.}
\label{fig:qcom}
\end{figure}

%

We define the quantum $A_1$ $T$-system for the variables
$T_{i,j}$ subject to the commutation relations \eqref{qcomzero} to be:
\begin{equation}\label{qtsys}
q\,T_{i,j+1}T_{i,j-1}=T_{i+1,j}T_{i-1,j}+1\qquad (i,j \in \Z)
\end{equation}

As in the commuting system \eqref{clatsys}, this is a ``three-term''
recursion relation in the variable $j$: All the variables $\{T_{i,j}\vert
i,j\in\Z,\  i\equiv j\! \!\mod 2\}$ are determined via these equations in
terms of the initial data $\bx_\bj=(T_{a,j_a})_{a\in\Z}$ with $|j_a-j_{a-1}|=1$.

Mutations $\mu_a^\pm$ are now implemented by using
the relations \eqref{qtsys} in the forward direction $T_{a,j_a}\to
T_{a,j_a+2}$ or backward direction $T_{a,j_a}\to T_{a,j_a-2}$.

Using Equations \eqref{qcomzero} and \eqref{qtsys}, the commutation
relations between cluster variables within the same seed $\bx_{\bj}$ are
determined for any $\bj$.
\begin{lemma}
Within each admissible initial data set of the quantum $A_1$ $T$-system, we have
the commutation relations (see Fig.\ref{fig:qcom}):
\begin{eqnarray}
T_{i-2k-m,j-m}T_{i,j}&=&q^{(-1)^k{1-(-1)^m\over 2}}  \,T_{i,j}T_{i-2k-m,j-m}\nonumber \\
T_{i+2k+m,j-m}T_{i,j}&=&q^{(-1)^k{1-(-1)^m\over 2}}  \,T_{i,j}T_{i+2k+m,j-m}\label{qcomm}
\end{eqnarray}
for all $i,j\in \Z$ and $k,m\in \Z_+$, with $i+j=0$ mod 2.
\end{lemma}
\begin{proof} By induction under mutation.  The relations
  \eqref{qcomm} reduce to \eqref{qcomzero} for $\bx_0$ under the
  specialization $j=m=1$.  

Suppose the variables in the admissible data set
  $\bx_\bm$ satisfy \eqref{qcomm} and that $m_{a-1}=m_{a+1}=m_a+1=m$
  for some fixed value of $a\in\Z$.
Let us apply \eqref{qtsys} to perform a
  mutation $\mu_a^+:\bm\to \bm'$ with $m'_b=m_b+2\delta_{b,a}$. We
  must check that all the commutation relations \eqref{qcomm} between
  any $(T_{b,m_b})$ for $b\neq a$ and $T_{a,m_a'}$ hold. Writing
  $i=a$, $j=m_a+1$, the new cluster variable is $T_{i,j+1}$, given by
  $qT_{i,j+1}T_{i,j-1}=T_{i+1,j}T_{i-1,j}+1$.  Let $k=b$, $\ell=m_b$
  for some $b\neq a$.  Without loss of generality, let's assume that
  $k=i+r+s$ and $\ell=j-r$, $r\geq 0,s\geq 1$.  Then by the
  commutation relations \eqref{qcomm}, we have
$$T_{i+r+s,j-r}(T_{i+1,j}T_{i-1,j})=(T_{i+1,j}T_{i-1,j})T_{i+r+s,j-r}$$
Henceforth, $T_{i+r+s,j-r}$ must commute with $T_{i,j+1}T_{i,j-1}$, and we obtain
$$T_{i+r+s,j-r}T_{i,j+1}=T_{i,j+1}T_{i+r,j-r-s} q^{-(-1)^{s}{1-(-1)^{r-1}\over 2}}$$
in agreement with \eqref{qcomm}. The Lemma follows.
\end{proof}
\begin{remark}
Note that $T_{i,j}$ and $T_{i',j'}$ in the same cluster $\bx_\bj$
commute if $j\equiv j'\mod 2$. 
From the definition of admissible data sets, we see that 
$T_{i,j}$and  $T_{i',j'}$ do not belong to the same cluster if 
$|i-i'|<|j-j'|$.
\end{remark}

\begin{remark}\label{barinv}
  Equation \eqref{qtsys} satisfies a ``bar invariance" property in the
  following sense.  Let $*$ denote an algebra antiautomorphism of
  $\cA$, where $q^*=q^{-1}$ and $T_{i,i\!\!\mod 2}^*=q T_{i,i\!\!\mod 2}$ for
  $i\in\Z$. Then $T_{i,j}^*=q T_{i,j}$ for all $i,j\in\Z$, $i+j=0$ mod
  2.  This result is obtained by conjugating the quantum $T$-system
  \eqref{qtsys} by $T_{i,j-1}$ and using the commutation relations
  \eqref{qcomm}.
\end{remark}

Using these commutation relations, we see that the $T$-system relation
\eqref{qtsys} is exactly of the form of the quantum cluster mutation
\eqref{qmutation}, upon the renormalization of variables $X_{i,j} =
q^{1/2} T_{i,j}$. We note that the subset of mutations \eqref{qtsys}
which we consider in the infinite rank cluster algebra makes sense,
because the product on the right hand side of a mutation has only a
finite number of factors (at most three).

A finite-rank quantum cluster algebra has a Laurent property
\cite{BZ}, that is, cluster variables are Laurent polynomials as
functions of any cluster seed,
as in the case of a commutative cluster algebra. In the quantum case,
the coefficients are in $\Z[q,q^{-1}]$. It is not completely
obvious that this carries over to the current case, which has
infinite rank.  However we will show that the solutions of the quantum
$A_1$ $T$-system have the Laurent property, by constructing
explicit formulas for the solutions $T_{i,j}$ of the $A_1$ quantum
$T$-system in terms of any initial data $x_\bj$. The coefficients are
in $\Z_+[q,q^{-1}]$, which is the analog of the positivity property of
cluster algebras \cite{FZI}.

\section{Quantum networks and the general solution}\label{section5}

Here, we generalize the results of \cite{DF} for the network solution
of the $T$-system in terms of arbitrary admissible initial data to the
non-commutative case. The solution of \eqref{qtsys} in terms of any given
admissible data $\bx_j$ is expressed as a quantum network partition function.

\subsection{$U$ and $V$ matrices}
Let $a,b$ be elements of $\mathcal A$. Define the 
matrices
\begin{equation}\label{dumat}
 U(a,b)=\begin{pmatrix} 1 & 0\\ q^{-1}b^{-1} & a b^{-1}\end{pmatrix}, \qquad 
V(a,b)=\begin{pmatrix}a b^{-1} & b^{-1}\\ 0 & 1\end{pmatrix}.
\end{equation}

These are interpreted as an elementary transfer matrix or ``chip",
along a lattice with two rows, going from left to right:
$U_{i,j}(a,b)$ or $V_{i,j}(a,b)$ is the weight of the edge connecting
the dot (entry connector) in row $i$ on the left to the dot (exit connector) in row $j$ on the right in
those elementary chips:
\begin{eqnarray}\label{udnet}
&&\epsfxsize=8cm \epsfbox{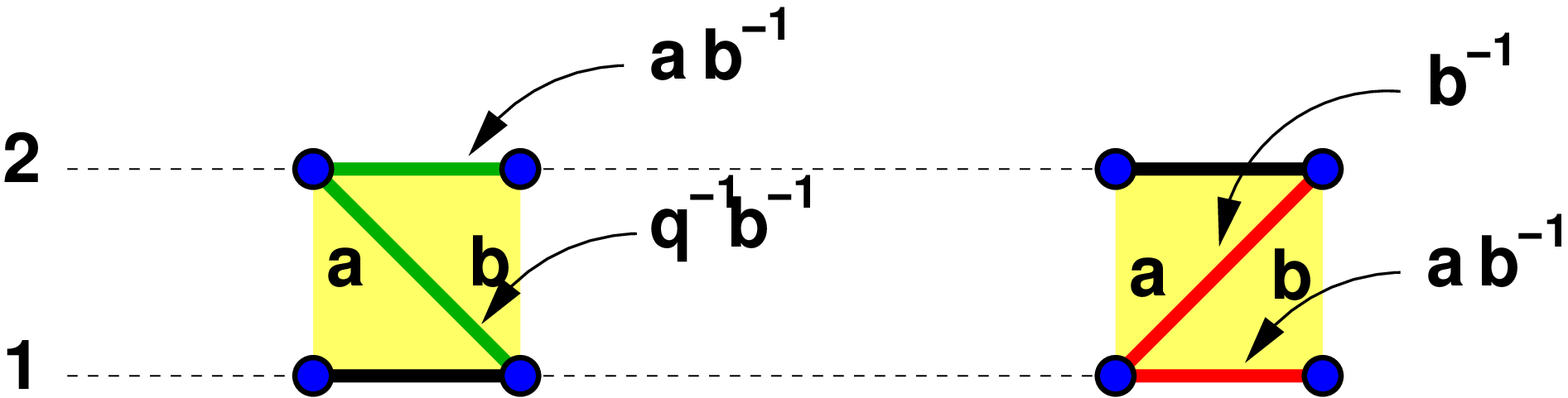} \nonumber \\
&& \qquad  \qquad  \   U(a,b) \qquad  \qquad  \qquad \qquad
V(a,b)
\end{eqnarray}
Note that we have represented the variables $a,b$ as attached
to the {\it faces} of the chips, separated by their edges.

A {\it quantum network} is obtained by the concatenation of such chips,
forming a chain where the exit connectors $1,2$ of each chip in the
chain are identified with the entry connectors of the next chip in the
chain, while face labels are well-defined. The latter condition
imposes that $U$ and $V$ arguments themselves form a chain
$a_1,a_2,...$, for instance:
\begin{equation}\label{duex}
W=V(a_1,a_2)V(a_2,a_3)U(a_3,a_4)V(a_4,a_5)U(a_5,a_6)
\end{equation}
corresponds to the network:
$$
\epsfxsize=7cm \epsfbox{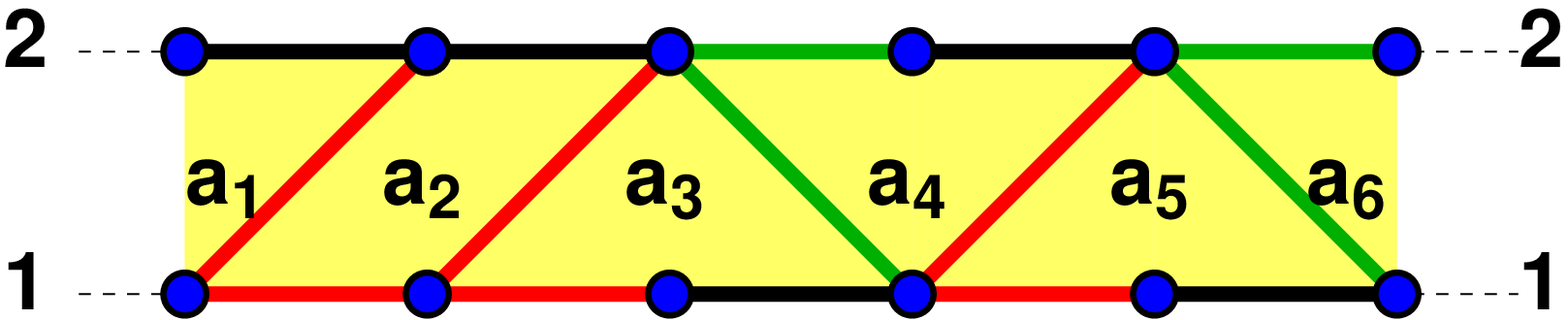}
$$
The partition function of a quantum network with matrix of weights $W$
with entry connector $i$ and exit connector $j$ is $W_{i,j}$.  It
the sum over paths from entry $i$ to exit $j$ of the product on
the edges, taken in the order they are traversed.


\begin{lemma}\label{prealem}
  Let $a,b,c\in \mathcal A$ be invertible elements with relations
  $ba=q ab$ , $bc=qcb$ and $ac=ca$ in $\mathcal A$, then
\begin{eqnarray} 
V(a,b)U(b,c)&=&U(a,b')V(b',c),
\end{eqnarray}
where $b'$ is defined by the relation $q\, b' b=ac+1$. This definition
implies that $cx=q b' c$ and $a x=q b' a$.
\end{lemma}
\begin{proof}
Direct calculation:
$$V(a,b)U(b,c)=
\begin{pmatrix} (a+c^{-1})b^{-1} & c^{-1}\\ q^{-1}c^{-1} & b
  c^{-1}\end{pmatrix}$$ and
$$U(a,b')V(b',c)
=\begin{pmatrix} {b'}c^{-1}& c^{-1}\\ q^{-1}c^{-1} & q^{-1}{b'}^{-1}
  c^{-1}+a {b'}^{-1}\end{pmatrix}$$ 
Setting the two expressions equal, we find $q {b'} b=1+a c$ from the
$(1,1)$ element, $q {b'} b= 1+q {b'} a {b'}^{-1} c$
from the $(2,2)$ one.  Since $a$
commutes with $c$, it commutes with $1+ac$, and the first identity
implies that $a{b'}b={b'}ba=q {b'} ab$, i.e.
$a{b'}=q{b'}a$, and $c{b'}=q {b'} c$. The Lemma follows.
\end{proof}

Let $\bx_{\bj}$ be admissible data. Then $j_i=j_{i+1}\pm1$. We
associate $V^{(i)}(\bx_\bj) = V(T_{i,j_i},T_{i+1,j_{i+1}})$ if $j_{i+1}-j_i=-1$ and
$U^{(i)}(\bx_\bj) = U(T_{i,j_i},T_{i+1,j_{i+1}})$ if $j_{i+1}-j_i=1$. Thus for the
boundary path $(i,j_i)_{i\in\Z}$ on the lattice, $V$ is associated
with a down step and $U$ is associated with an up step.

Therefore Lemma \ref{prealem} has a graphical representation in terms
of a local mutation of admissible data,
$$
\raise -.5truecm \hbox{${\epsfxsize 2.cm \epsfbox{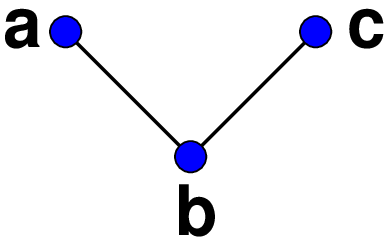}}$}=
\raise -.5truecm \hbox{${\epsfxsize 2.cm \epsfbox{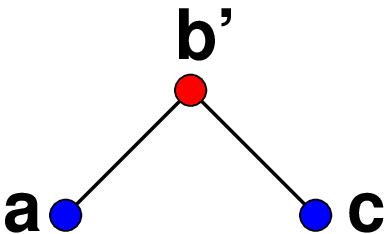}}$} \nonumber \\
$$
In other words, Lemma \ref{prealem} is an implementation of a
mutation of an admissible boundary, using the $A_1$ $T$-system. That
is, we have the relation
\begin{equation}\label{mutmat}
V^{(i-1)}(\bx_{\bj})U^{(i)}(\bx_\bj) = U^{(i-1)}(\bx_{\bj'})V^{(i)}(\bx_{\bj'})
\end{equation}
where $\bx_{\bj'}=\mu_i^+(\bx_\bj)$.


\subsection{Main Theorem}

Let $(i,j)\in\Z^2$ ($i+j=0 \mod 2$) above a fixed
admissible boundary $\bj$, that is, with $j\geq j_i$.

\begin{figure}
\centering
\includegraphics[width=12.cm]{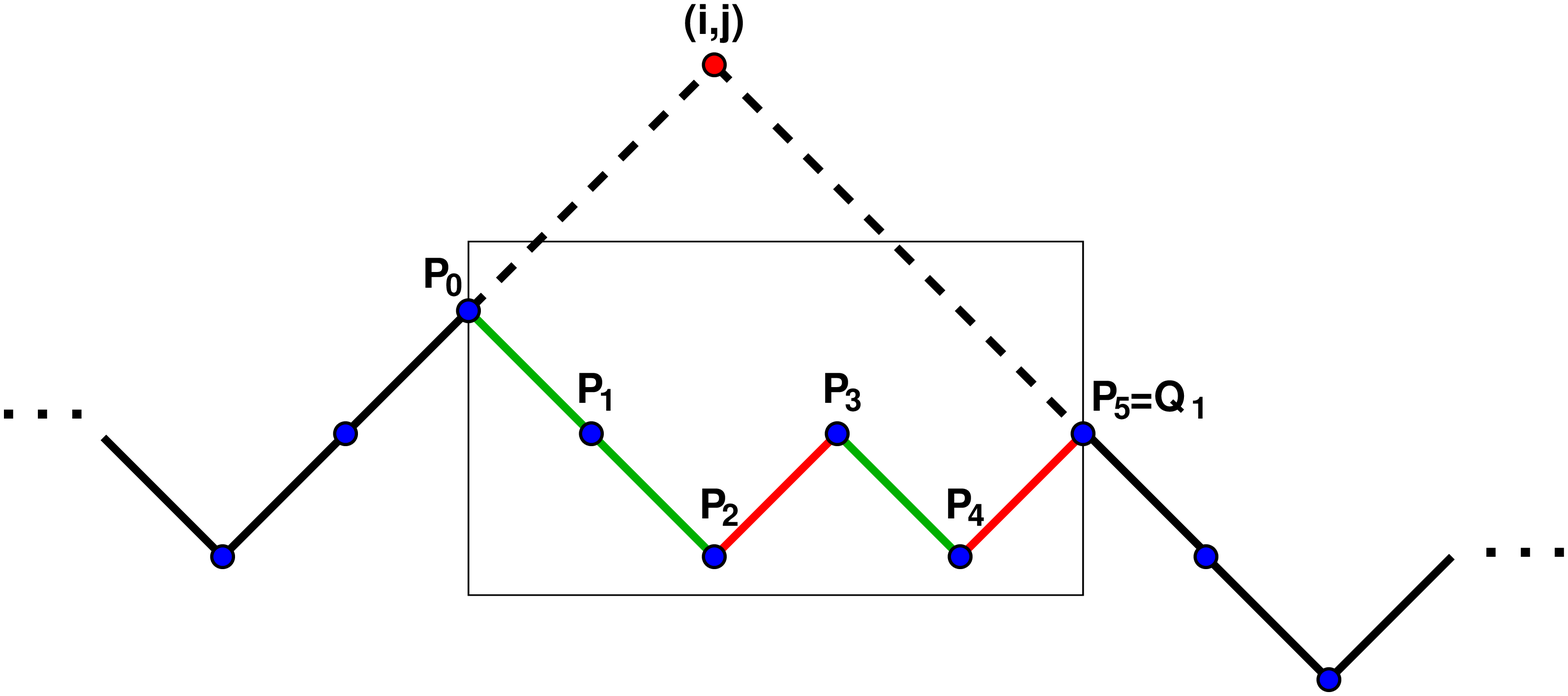}
\caption{\small The projection of a point $(i,j)$ onto a given boundary $s_\bm$.}
\label{fig:proj}
\end{figure}

\begin{defn}
  The projection of the point $(i,j)$ onto the boundary $\bj$ is the
  set of points
  $(i_0,j_{i_0}),(i_0+1,j_{i_0+1}),\ldots,(i_1,j_{i_1})$, the portion
  of boundary between the lines $(i+k,j+k)_{k\in\Z}$ and
  $(i+k,j-k)_{k\in\Z}$, with endpoints $P_0=(i_0,j_{i_0})$ and
  $Q_1=(i_1,j_{i_1})$ such that $j_{i_0}-i_0=j-i$ with
  $i_0$ maximal and $j_{i_1}+i_1=j+i$ with $i_1$ minimal.
\end{defn}

Figure \ref{fig:proj} is an example of such a projection. It is a path
along the boundary points $(i,j_i)$ from
the vertex $P_0$ to
the vertex $Q_1$ formed by a succession of down (SE) steps
$d=(1,-1)$ and up (NE) steps $u=(1,1)$. By definition such a path,
if non-empty, starts with a down step and ends up with an up step.

To any path $p=(P_0,P_1,P_2,...,P_n=Q_1)$ 
made of steps $S_k= P_k-P_{k-1}\in \{ d,u\}$, $k=1,2,...,n$,
we associate a matrix product as follows.
We define $M_k(d,p)=V(T_{P_{k-1}},T_{P_k})=V^{(i_{k-1})}(\bx_\bj)$
and $M_k(u,p)=U(T_{P_{k-1}},T_{P_k})=U^{(i_{k-1})}(\bx_\bj)$, with the matrices
$V$ and $U$ as in \eqref{dumat} and 
where for any point of the form $P=(x,y)$ we denote by $T_P:=T_{x,y}$.
Define
\begin{equation}\label{matprod}
M(p)=M_1(S_1,p)M_2(S_2,p) \cdots M_n(S_n,p).
\end{equation}
This product is the weight matrix of the network made up of a concatenation
of the basic network chips of the form \eqref{udnet} determined by $p$.
Let $N(p)$ be the corresponding quantum network.

\begin{example}
  The quantum network $N(p)$ in Equation \eqref{duex} with weight matrix
  $M(p)=W$ corresponds to the path
$$ p= \quad \raise -1.truecm \hbox{${\epsfxsize=5cm \epsfbox{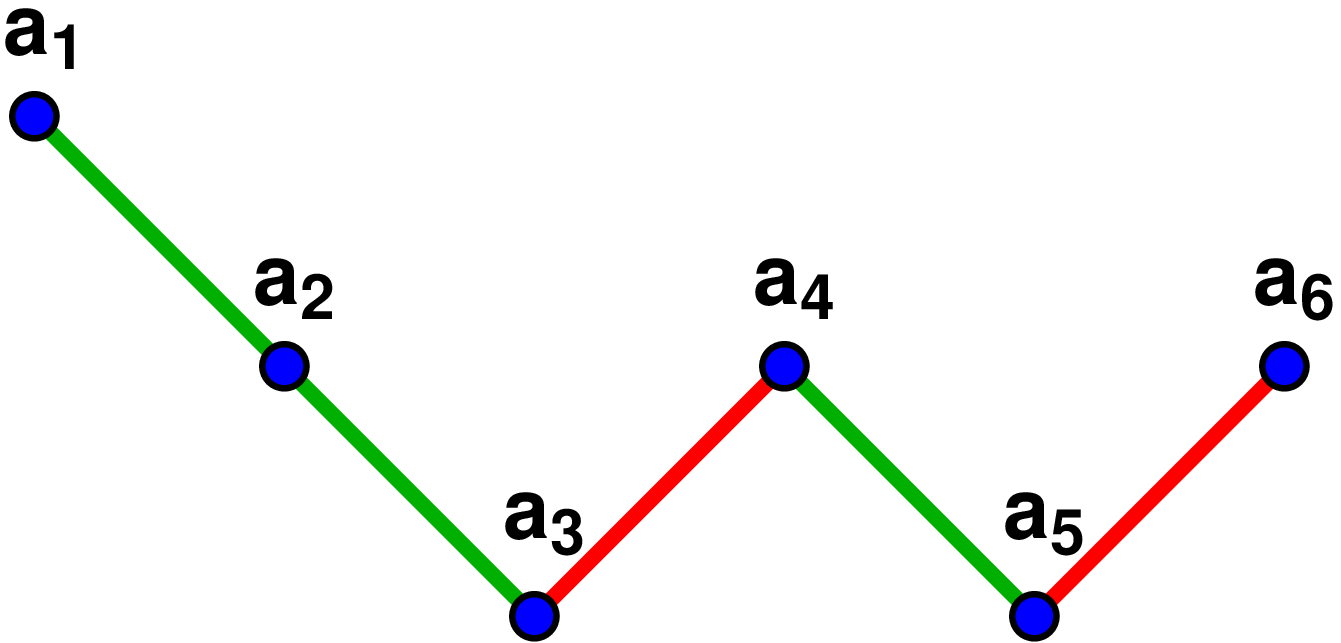}}$} $$
made of a succession of steps $ddudu$, and with a set of vertices of
the form $P_{i-1}=(x_i,y_i)\in \Z^2$ with $x_i+y_i=0$ mod 2,
$i=1,2,...,6$, with $a_i=T_{x_i,y_i}$.
\end{example}

 The main theorem of this section is the following:

\begin{thm}\label{solq}
Let $p$ be the projection of $(i,j)$ onto the boundary $\bj$, with endpoints
$(i_0,j_0)\equiv (i_0,j_{i_0})$ and $(i_1,j_1)\equiv (i_1,j_{i_1})$.
As a function of the admissible data set $\bx_{\bj}$,
\begin{equation}\label{qsol}
T_{i,j}=M(p)_{1,1} T_{i_1,j_1}
\end{equation}
\end{thm}
\begin{proof}
  This is proved by induction under mutations of initial data. Let
  $\bx_\bj$ be some initial data whose boundary contains the
  point $(i,j)$. For such a case, we have
  $(i,j)=(i_0,j_{0})=(i_1,j_{1})$, $T_{i,j}=1\times
  T_{i,j}=({\mathbf I})_{1,1} T_{i_1,j_1}$, and \eqref{qsol} is
  trivially satisfied.

  Assume \eqref{qsol} holds for some boundary $\bx_\bj$, let us show it
  also holds for the boundary $\bx_{\bj'}$ with $\bj'=\mu_a^\pm(\bj)$,
  that is, $j_a$ is changed for one value of $a\in\Z$, and all other
  values of $j_i$ remain unchanged.

  If $a>i_1$ or $a<i_0$, then the mutation does not
  affect the formula \eqref{qsol}, as the boundary $\bj'$
  is modified outside of the projection of $(i,j)$ onto it, whereas
  $M(p)$ and hence $T_{i,j}$ only depends on the boundary values within
  the projection.

\begin{figure}
\centering
\includegraphics[width=16.cm]{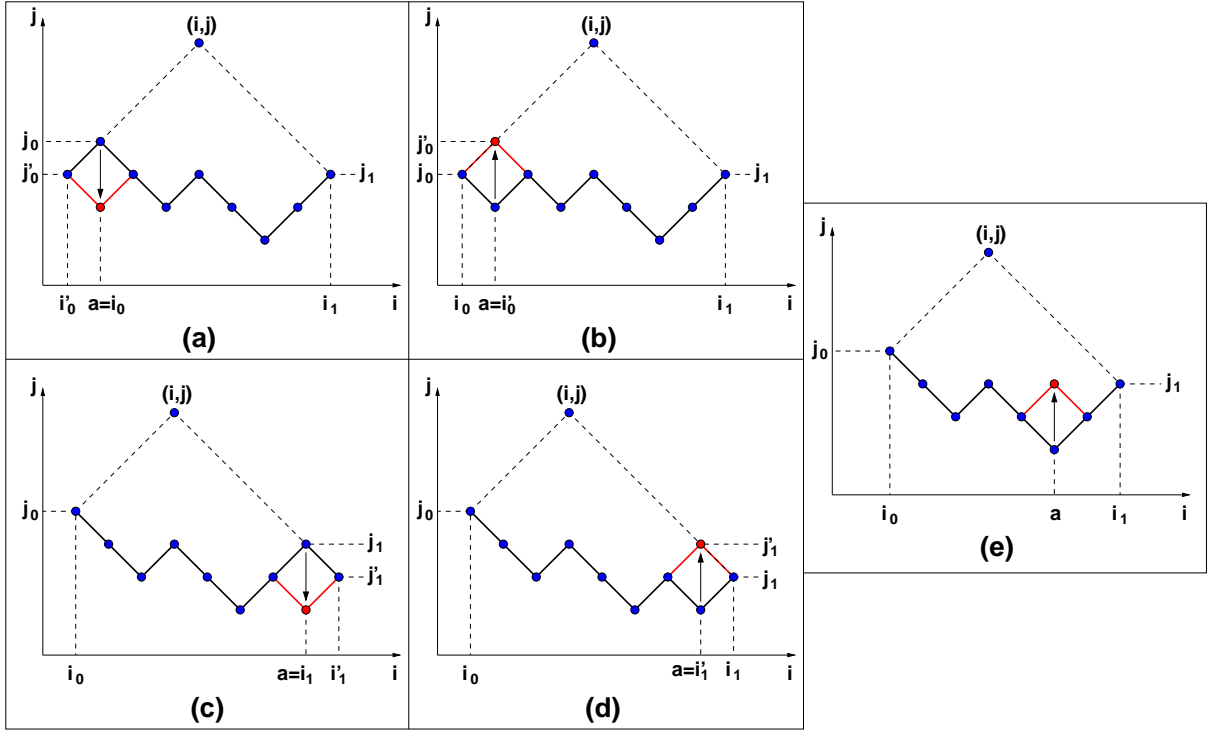}
\caption{\small The five cases to be considered in the proof of Theorem \ref{solq}:
(a) $a=i_0$ (b) $a=i_0+1$ (c) $a=i_1$ (d) $a=i_1-1$ (e) $i_0+1<a<i_1-1$. For each case,
we indicate the mutation by an arrow. The projection of $(i,j)$ onto the boundary is
modified by the mutation only in the first four cases.}
\label{fig:fivecases}
\end{figure}

If $i_0\leq a \leq i_1$, five situations may occur, as sketched in
Fig.\ref{fig:fivecases} (a-e).  Let $p$ the projection of
$(i,j)$ onto $\bj$, and $p'$ the projection of $(i,j)$ onto $\bj'$, and $j_0'\equiv j_{i_0'}$, $j_1'\equiv j_{i_1'}$.
\begin{description}
\item[(a)] If $a=i_0$ and $\mu=\mu_{a}^-$, then
  $(i_0',j_{0}')=(i_0-1,j_{0}-1)$. The first step
  of $p$ is must be $d$ for this mutation to be one of the $T$-system
  equations.  Separating out the contribution of this first step of $p$,
  we write
  $M(p)=V^{(i_0)}(\bj){\tilde M}(p)$.  Using the fact that
  $(U)_{1,j}=\delta_{j,1}$:
\begin{eqnarray*}M(p)_{1,1}&=&(U^{(i_0-1)}(\bx_\bj) M(p))_{1,1}\\
&=&(U^{(i_0-1)}(\bx_\bj)V^{(i_0)}(\bx_\bj){\tilde M}(p))_{1,1}\\
&=&(V^{(i_0-1)}(\bx_{\bj'})U^{(i_0)}(\bx_{\bj'}){\tilde M}(p))_{1,1}=M(p')_{1,1}
\end{eqnarray*}
where in the last line, we applied Eq. \eqref{mutmat}.
Equation \eqref{qsol} follows, as $T_{i_1,j_{1}}=T_{i'_1,j'_{1}}$.

\item[(b)] If $a=i_0+1$ and the first two steps of $p$
  are $d$, $u$, then $\mu=\mu_{a}^+$, $(i_0',j_{0}')=(i_0+1,j_{0}+1)$.
  Therefore,
\begin{eqnarray*}
M(p)_{1,1}&=&(V^{(i_0)}(\bx_\bj)U^{(i_0+1)}(\bx_\bj){\tilde M}(p))_{1,1}\\
&=&(U^{(i_0)}(\bx_{\bj'})V^{(i_0+1)}(\bx_{\bj'}){\tilde M}(p))_{1,1}=M(p')_{1,1}
\end{eqnarray*}
by application of Eq. \eqref{mutmat}. Equation \eqref{qsol}
follows, as
$T_{i_1,j_{1}}$ is unchanged by the mutation.

\item[(c)] If $a=i_1$ and $\mu=\mu_{a}^-$, then
  $(i_1',j_1')=(i_1+1,j_1-1)$. 
The last step of $p$ is $u$, so 
$M(p)={\tilde M}(p) U^{(i_1-1)}(\bx_\bj)$. 
Since $(V(a,b))_{i,1}=\delta_{i,1}a b^{-1}$, 
\begin{eqnarray*}
M(p)_{1,1} T_{i_1,j_1}&=&
\big( {\tilde M}(p) U^{(i_1-1)}(\bx_\bj)V^{(i_1)}(\bx_{\bj})\big)_{1,1} 
T_{i_1+1,j_1-1}\\
&=&
\big( {\tilde M}(p) V^{(i_1-1)}(\bx_{\bj'}) U^{(i_1)}(\bx_{\bj'})\big)_{1,1} 
T_{i_1+1,j_1-1}\\
&=&M(p')_{1,1}T_{i_1',j'_1}.
\end{eqnarray*}

\item[(d)] If $a=i_1-1$ and $\mu=\mu_{a}^+$,  then
  $(i_1',j_1')=(i_1-1,j_1+1)$ and the last two steps of $p$ are $d$, $u$.
Therefore,
\begin{eqnarray*}
M(p)_{1,1}T_{i_1,j_1}&=&({\tilde M}(p)V^{(i_1-2)}(\bx_\bj) 
U^{(i_1-1)}(\bx_{\bj}))_{1,1}T_{i_1,j_1}\\
&=&({\tilde M}(p)U^{(i_1-2)}(\bx_{\bj'})V^{(i_1-1)}(\bx_{\bj'}))_{1,1}T_{i_1,j_1}\\
&=&({\tilde M}(p)U^{(i_1-2)}(\bx_{\bj'}))_{1,1}T_{i_1-1,j_1+1}=M(p')_{1,1}T_{i_1',j_1'}
\end{eqnarray*}
again using $(V(a,b))_{i,1}=\delta_{i,1}a b^{-1}$. 

\item[(e)] If $i_0+1<a<i_1-1$,  the endpoints of the projection onto
  the boundary do not
  change and the mutation $\mu_a^\pm$ amounts to a change of ordering
  of one pair of factors within the product $M(p)$ of the form
  \eqref{mutmat}, which corresponds to writing it as $M(p')$,
  $p'=\mu_a^\pm(p)$ and \eqref{qsol} follows.
\end{description}
\end{proof}

Theorem \ref{solq} may be rephrased in the language of a quantum
network partition function:
\begin{cor}
  The quantity $T_{i,j}T_{i_1,j_1}^{-1}$ is the partition function of
  the quantum network $N(p)$, with weight matrix $M(p)$, with entry
  and exit connector $1$, where $p$ is the projection of $(i,j)$ onto
  the boundary.
\end{cor}

As all weights involved in the network $N(p)$ are Laurent monomials of
the initial data $\bx_{\bj}$ with coefficients in $\Z_+[q,q^{-1}]$, we deduce the
following positivity result, the quantum version of the
Fomin-Zelevinsky positivity conjecture:

\begin{cor}
  The expression for $T_{i,j}$ in terms of any initial data $x_\bj$ is
  a Laurent polynomial with coefficients in $\Z_+[q,q^{-1}]$.
\end{cor}

\begin{example}
  Consider the case where $(i,j)=(0,4)$, with a boundary projection of
  the form: $p=((-2,2),(-1,1),(0,0),(1,1),(2,2))$. Then
  $(i_0,j_0)=(-2,2)$,
  $(i_1,j_1)=(2,2)$.  It consists of the steps $dduu$, hence
$$M(p)=V^{(-2)}(\bx_\bj)V^{(-1)}(\bx_\bj)U^{(0)}(\bx_\bj)U^{(1)}(\bx_\bj)$$
which yields
$$T_{0,4}=M(p)_{1,1}T_{2,2}= T_{-2,2}T_{0,0}^{-1}T_{2,2}+q^{-1}(T_{-2,2}T_{-1,1}^{-1}T_{0,0}^{-1}
+T_{-1,1}^{-1})(T_{1,1}^{-1}T_{2,2}+T_{0,0}T_{1,1}^{-1}) $$
The five monomials forming $T_{0,4}T_{2,2}^{-1}$ are the weights of the five paths $1\to 1$ in the
network $N(p)$ with weight matrix $M(p)$:
$$ \epsfxsize=8cm \epsfbox{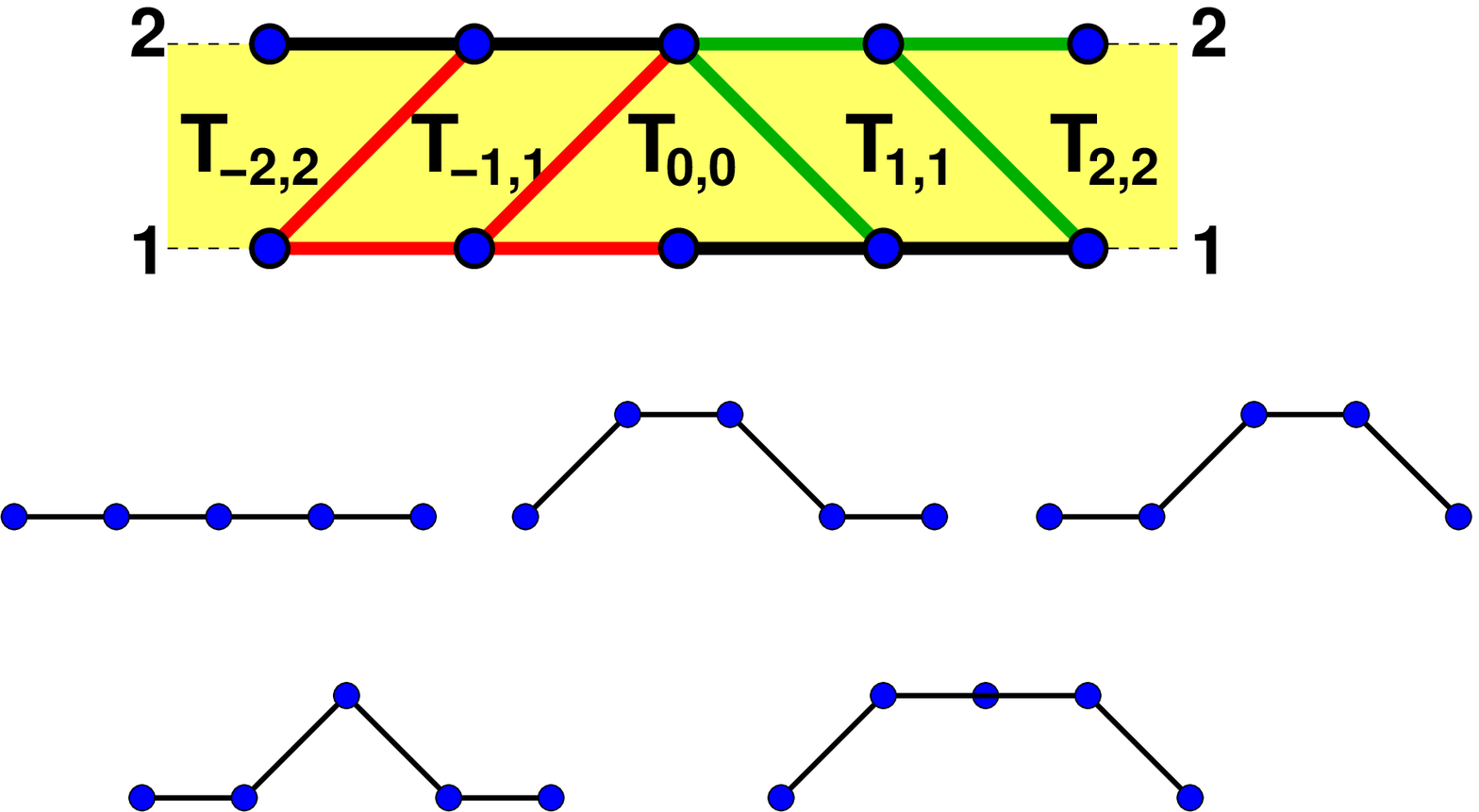}$$
\end{example}

\begin{figure}
\centering
\includegraphics[width=14.cm]{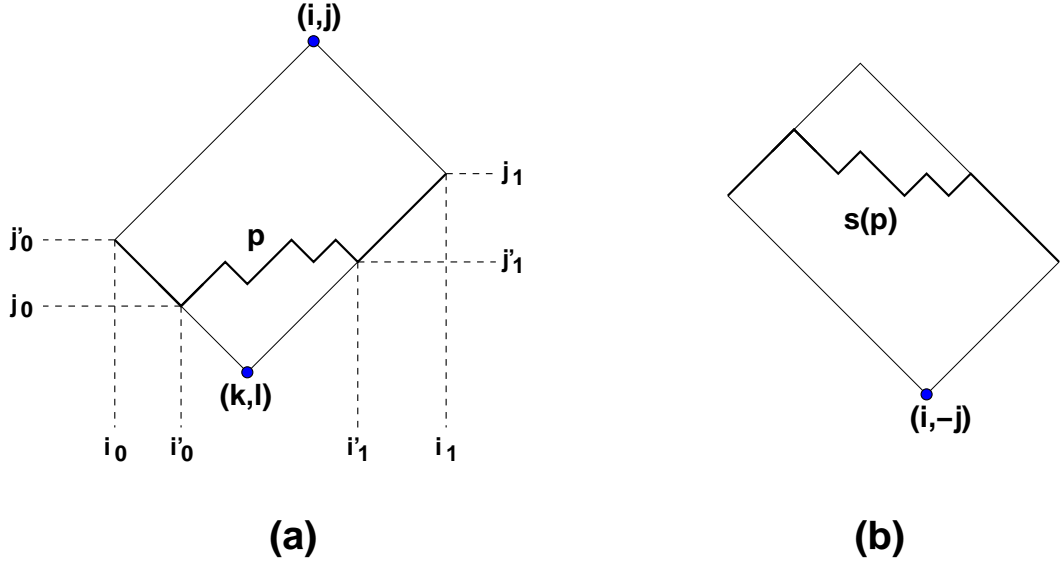}
\caption{\small (a) The projection of a point $(i,j)$ onto a boundary $\bj$,
and the corresponding point $(k,l)$. (b) The action of the reflection $s:(i,j)\mapsto (i,-j)$.}
\label{fig:bigrect}
\end{figure}

Finally, we can use the symmetry of the $T$-system under the bar
involution to compute $T_{k,\ell}$ with $\ell\leq j_k$. Given a
boundary $\bj$ and a point $(i,j)$ above it with $j\geq j_i$, let
$p$ denote the projection of $(i,j)$ onto the boundary with endpoints
$(i_0,j_0)$ and $(i_1,j_1)$. Then
$$i={i_0+i_1+j_1-j_0\over 2}, \qquad j={j_0+j_1+i_1-i_0\over 2}.$$
Let $(k,l)$ be the point under the boundary 
such that
\begin{equation}\label{defkl} 
k={i_0+i_1+j_0-j_1\over 2}, \qquad l={j_0+j_1+i_0-i_1\over 2}.\end{equation}
(See Fig.\ref{fig:bigrect} (a) for an illustration.).
Let $s: (i,j)\mapsto (i,-j)$ denote a reflection. Under $s$, the boundary $\bj$
is sent to $-\bj$ (see Fig.\ref{fig:bigrect} (b)).
The
projection of $s(k,l)$ onto the reflected boundary $-\bj$ is $s(p')$,
a sub-path of $s(p)$, with endpoints $s(i_0',j_0')$ and $s(i_1',j_1')$, with
$i_0-i_0'=j_0-j_0'\geq 0$ and $i_1'-i_1=j_1-j_1'\geq 0$. Note
that $s(p)=u^{i_0-i_0'} s(p' )d^{i_1'-i_1}$.
We have the following:

\begin{thm}\label{identwo}
Let $\bj$, $(i,j)$, $p$ and $M(p)$ be as in Theorem \ref{solq}, and $(k,l)$ the
point under the boundary $\bj$ defined by \eqref{defkl}. In terms of the initial
data $\bx_\bj$, we have:
\begin{equation} 
T_{k,l}= \big( M(p)\big)_{2,2}  T_{i_1,j_1}.
\end{equation}
\end{thm}
\begin{proof}
  Note first that $S_{i,j}=qT_{i,-j}$ is a solution of the $A_1$ quantum
  $T$-system with $q\to q^{-1}$. Indeed, from \eqref{qtsys},
  $q^{-1}S_{i,j-1}S_{i,j+1}=q^{-2} S_{i+1,j}S_{i-1,j}+1$, and upon
  applying the bar involution of Remark \ref{barinv},
  $q^{-1}S_{i,j+1}S_{i,j-1}=S_{i+1,j}S_{i-1,j}+1$.  Therefore with the
  initial data $q\bx_{-\bj}=(q T_{i,j_i})_{i\in\Z}$, we have
$$q^{-1}S_{i,j}=T_{i,-j}=M(p,q\bx_{-\bj},q^{-1})_{1,1} q^{-1}S_{i_1,j_1}
=M(p,q\bx_{-\bj},q^{-1})_{1,1} T_{i_1,-j_1}$$ Here, we have made
explicit the arguments of $M(p):=M(p,\by,t)$, where $\by$ is the
boundary data in the $V,U$ matrices, and
the quantum parameter is $t$.
Similarly, let $U^{(i)}(\by)=U^{(i)}(\by,t)$
and $V^{(i)}(\by)=V^{(i)}(\by,t)$ where $t$ is the quantum parameter.
Then $M(p,q\bx_{-\bj},q^{-1})$ is obtained from $M(p,\bx_\bj,q)$ upon
substitution of the matrices
$$ U^{(i)}(\bx_\bj,q)\mapsto U^{(i)}(q\bx_{-\bj},q^{-1})\quad {\rm and} \quad
 V^{(i)}(\bx_\bj,q) \mapsto  V^{(i)}(q\bx_{-\bj},q^{-1}) $$
We note that
$$ U(q a,q b;q^{-1}) =J V(a,b;q)J \quad {\rm and}\quad V(q a,q b;q^{-1}) =J U(a,b;q)J $$
where $J=\begin{pmatrix} 0 & 1\\ 1 &
  0 \end{pmatrix}$ is the permutation matrix. Therefore,
$$T_{i,-j}=(JM(s(p),\bx_{-\bj},q)J)_{1,1} T_{i_1,-j_1}=M(s(p),\bx_{-\bj},q)_{2,2}T_{i_1,-j_1}$$
where $s(p)$ is the reflected path with $d$ and $u$ steps interchanged. This last identity
corresponds to the reflecting of the entire initial picture. Upon renaming $(i,-j)\to (k,l)$,
$(i_1,-j_1)\to (i_1',j_1')$, $(i_0,-j_0)\to (i_0',j_0')$, 
$s(p)\to p'$, and $\bx_{-\bj}\to \bx_\bj$, we deduce that
$$T_{k,l}=M(p')_{2,2} T_{i_1',j_1'}$$
Recalling finally that $p=d^{i_0-i_0'} p' u^{i_1'-i_1}$
and that, due to the particular triangular form of the $U,V$ matrices:
$$\left(\prod_{i=i_0'}^{i_0-1}V^{(i)}(\bx_\bj)\right)_{2,m}=\delta_{m,2},\qquad 
\left(\prod_{i=i_1}^{i_1'-1} U^{(i)}(\bx_\bj)\right)_{m,2}=\delta_{m,2} T_{i_1',j_1'}T_{i_1,j_1}^{-1},$$
the Theorem follows.
\end{proof}

\subsection{Conserved quantities and $V, U$ matrices}

In this section we investigate the content of the full network matrix
$M(p)$ of Theorems \ref{solq} and \ref{identwo}, in the special case of the
fundamental boundary $\bj_0$, with heights $j_i=i\mod2$.

\begin{figure}
\centering
\includegraphics[width=10.cm]{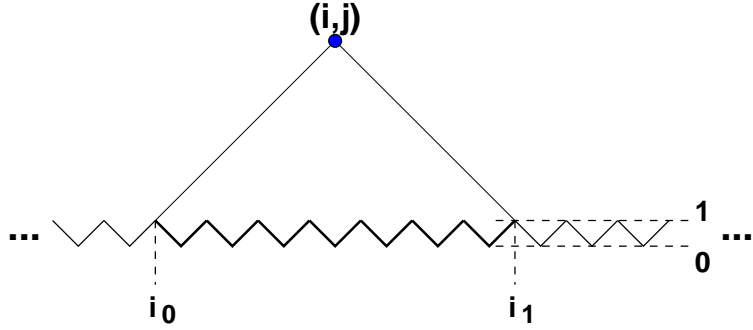}
\caption{\small The projection of a point $(i,j)$ onto the fundamental boundary $\bj_0$.}
\label{fig:bigsquare}
\end{figure}

Let $p$ be the projection of $(i,j)\in \Z^2$ with $j\geq (i\
{\rm mod}\ 2)$ onto the boundary $\bj_0$, consisting of vertices $\{(a,j_a)\,:
\,i_0\leq a \leq i_1\}$ with $i_0,i_1$ odd (see
Fig.\ref{fig:bigsquare}).

We have
$$ M(p)=(V^{(i_0)}U^{(i_0+1)}
\cdots V^{(i_1-2)}U^{(i_1-1)})(\bx_\bj).$$ 
Let $n=(i_1-i_0)/2$, then
the length of $p$ is $2n$.
%

We focus on the non-diagonal terms of $M(p) T_{i_1,1}$. These involve
what we call (in analogy with the commutative case) the conserved
quantities of the quantum $A_1$ $T$-system:
\begin{lemma}
Let $(i,j)\in \Z^2$, and $T$ a solution of the quantum $A_1$ $T$-system
\eqref{qtsys}. Then
$$c_{i,j}=T_{i-1,j+1}T_{i,j}^{-1}+T_{i,j}^{-1}T_{i+1,j-1}$$
is independent of $i-j$ and 
$$
d_{i,j}= T_{i-1,j-1}T_{i,j}^{-1}+T_{i,j}^{-1}T_{i+1,j+1}$$
is independent of $i+j$. That is,
$$c_{i,j}=c_{i-1,j-1}=c_{i-j,0}:=c_{i-j}\qquad d_{i,j}=d_{i+1,j-1}=d_{i+j,0}:=d_{i+j}$$
\end{lemma}
\begin{proof}
We write the $T$-system equations:
\begin{equation*}\left\{\begin{matrix}
qT_{i,j}T_{i,j-2}=T_{i+1,j-1}T_{i-1,j-1}+1 \hfill &\Rightarrow 
T_{i-1,j-1}^{-1}T_{i,j-2}=T_{i,j}^{-1}T_{i+1,j-1}+T_{i,j}^{-1}T_{i-1,j-1}^{-1}\hfill \\
qT_{i-1,j+1}T_{i-1,j-1}=T_{i,j}T_{i-2,j}+1 \hfill &\Rightarrow 
T_{i-1,j+1}T_{i,j}^{-1}=T_{i-2,j}T_{i-1,j-1}^{-1}+T_{i,j}^{-1}T_{i-1,j-1}^{-1}
\hfill \end{matrix}\right.
\end{equation*}
where in the first line we have used the $q$-commutation of $T_{i-1,j-1}^{-1}$ and $T_{i,j-2}$
and that of $T_{i-1,j+1}$ and $T_{i,j}^{-1}$ in the second.
Subtracting these last two equations leads to $c_{i,j}=c_{i-1,j-1}$. The
conservation of $d$ is proved in a similar manner.
\end{proof}

We also define by induction the following polynomials of the conserved
quantities:
\begin{equation*}\left\{ \begin{matrix} \varphi^{(-1)}_m=0, & \varphi^{(0)}_m=1, &
\varphi^{(p)}_m= \varphi^{(p-1)}_m c_{m+2p-2} -q \, \varphi^{(p-2)}_m\hfill \\
\theta^{(-1)}_m=0, & \theta^{(0)}_m=1, &
\theta^{(p)}_m=d_{m+2-2p}\, \theta^{(p-1)}_m- q^{-1}\theta^{(p-2)}_m \hfill
\end{matrix}\right. \quad (p\geq 1, m\in\Z)
\end{equation*}

Note that, while it is true that $[c_m,d_p]=0=[\varphi^{(i)}_m,\theta^{(j)}_p]$,
for all $i,j,m,p$, 
neither the $c$'s nor the $d$'s commute among themselves. 

We have:
\begin{thm}
Let $(i,j)$ be a point above the boundary $\bj=\bj_0$ and $p$ be the projection
of $(i,j)$ onto the boundary, $p$ a zig-zag path with endpoints
$(i_0,1)$ and $(i_1,1)$ and length $2n=i_1-i_0$, as in Fig.\ref{fig:bigsquare}. Then
\begin{eqnarray*}
 \big( M(p)\big)_{1,2} \, T_{i_1,1}&=&\varphi^{(n-1)}_{i-j+2},\\
 \big( M(p)\big)_{2,1} \, T_{i_1,1}&=&q^{-1}\, \theta^{(n-1)}_{i+j-2}.
\end{eqnarray*}

\end{thm}
\begin{proof} By induction on $n$. For $n=0$, $M(p)={\mathbf I}$, and
  the theorem holds, as $\varphi^{(-1)}_m=\theta^{(-1)}_m=0$. 

  Assume the theorem holds for paths $p'$ of length $2 n$.  Let
  $p$ be a path of length $2n+2$, with $i_0=i-n-1$,
  $i_1=i+n+1$, $j=n+2$. Denote by $p'$ the truncated projection
  between the lines $i=i_0$ and $i=i_1-2$.  For simplicity, we
  introduce the following notation: $U^{(i)}(\bx_\bj)=\langle i \vert
  U \vert i+1 \rangle$ and $V^{(i)}(\bx_\bj)=\langle i \vert V \vert
  i+1 \rangle$, with $\vert i \rangle\langle i \vert =I$ for all $i$.
  We have
\begin{eqnarray}
(M(p))_{1,2}T_{i_1,1}&=&\big(\langle i_0|(VU)^{n+1}|i_1\rangle\big)_{1,2}T_{i_1,1}\nonumber \\
&=&\big(\langle i_0|(VU)^n|i_1-2\rangle\langle i_1-2|VU|i_1\rangle)\big)_{1,2}T_{i_1,1}\nonumber  \\
&=&\big(M(p')\langle i_1-2|V|i_1-1\rangle)\big)_{1,2}T_{i_1-1,0}\nonumber \\
&=&(M(p'))_{1,2}T_{i_1-2,1}T_{i_1-2,1}^{-1}T_{i_1-1,0}+(M(p'))_{1,1}\nonumber \\
&=&\varphi^{(n-1)}_{i-j+2}T_{i_1-2,1}^{-1}T_{i_1-1,0}+T_{i-1,j-1} T_{i_1-2,1}^{-1} \label{interca}
\end{eqnarray}
by applying Theorem \ref{solq}. Repeating this calculation for 
$(M(p'))_{1,2}T_{i_1-2,1}=\varphi^{(n-1)}_{i-j+2}$ by the recursion hypothesis, we
get analogously:
$$ \varphi^{(n-1)}_{i-j+2} =\varphi^{(n-2)}_{i-j+2}T_{i_1-4,1}^{-1}T_{i_1-3,0}+T_{i-2,j-2} T_{i_1-4,1}^{-1}$$
We may consider this identity with indices $i,j$ shifted by $+1$, while $n$ remains fixed, namely
$$\varphi^{(n-1)}_{i-j+2} =\varphi^{(n-2)}_{i-j+2}T_{i_1-3,2}^{-1}T_{i_1-2,1}+T_{i-1,j-1} T_{i_1-3,2}^{-1}$$
from which we deduce:
\begin{eqnarray*}
\varphi^{(n-1)}_{i-j+2} T_{i_1-3,2}T_{i_1-2,1}^{-1}&=&
\varphi^{(n-2)}_{i-j+2}T_{i_1-3,2}^{-1}T_{i_1-2,1}T_{i_1-3,2}T_{i_1-2,1}^{-1}+T_{i-1,j-1} T_{i_1-2,1}^{-1}\\
&=& q\, \varphi^{(n-2)}_{i-j+2}+T_{i-1,j-1} T_{i_1-2,1}^{-1}
\end{eqnarray*}
by use of the $q$ commutation of $T_{i_1-2,1}$ and $T_{i_1-3,2}$.
Comparing with \eqref{interca}, we finally get
\begin{eqnarray*}(M(p))_{1,2}T_{i_1,1}&=&
\varphi^{(n-1)}_{i-j+2}(T_{i_1-3,2}T_{i_1-2,1}^{-1}+T_{i_1-2,1}^{-1}T_{i_1-1,0})-q\, \varphi^{(n-2)}_{i-j+2}\\
&=&
\varphi^{(n-1)}_{i-j+2}c_{i+n-2} -q\, \varphi^{(n-2)}_{i-j+2}=
\varphi^{(n-1)}_{i-j+2}c_{i-j+2n} -q\, \varphi^{(n-2)}_{i-j+2}=\varphi^{(n)}_{i-j+2}
\end{eqnarray*}
by the defining recursion relation for $\varphi^{(n)}_m$. The second statement of the theorem
follows analogously.
\end{proof}

\section{Quantum $Q$-system for $A_1$ and its fully non-commutative version}\label{section6}



\subsection{Quantum $Q$-system: from the path solution to the network solution}

\subsubsection{$A_1$ $Q$-system}
Closely related to the $A_1$ $T$-system is the $A_1$ $Q$-system:
\begin{equation}\label{aoneQsys} R_{j+1}R_{j-1}=R_j^2+1 \qquad (j\in \Z)
\end{equation}
which is satisfied by the 2-periodic solutions of the $A_1$ $T$-system in $i$,
namely with $T_{i+2,j}=T_{i,j}=T_{j\, {\rm mod}\, 2,j}=R_j$. The admissible initial data are
of the form $\bx_n=(R_n,R_{n+1})$ and are the restrictions of the $2$-periodic initial data
of the $T$-system with $\bj$ such that $j_i=n$ if $i-n=0$ mod 2, and $j_i=n+1$ otherwise.

The $A_1$ $Q$-system \eqref{aoneQsys} has an associated cluster algebra \cite{Ke07}
with fundamental seed made of the cluster $\bx_0=(R_0,R_1)$ and of the exchange
matrix $B_0=\begin{pmatrix} 0 & 2\\ -2 & 0\end{pmatrix}$.
A forward mutation $\mu^+:\bx_n\mapsto \bx_{n+1}$ takes
$R_n\to R_{n+2}=(R_{n+1}^2+1)/R_n$, while a backward one $\mu^-:\bx_n\mapsto \bx_{n-1}$ takes 
$R_{n+1}\to R_{n-1}=(R_n^2+1)/R_{n+1}$.

\subsubsection{Quantum $A_1$ $Q$-system from quantum cluster algebra}

The quantum cluster algebra associated to the cluster
algebra of the $A_1$ $Q$-system \eqref{aoneQsys}, is obtained by taking 
the admissible pair $(B_0,\Lambda)$ with $\Lambda=2(B_0)^{-1}=\begin{pmatrix} 0 & 1\\ -1 & 0\end{pmatrix}$
(see Ref.\cite{BZ,DFK10}), which amounts to the commutation relation $R_0R_1=q\,R_1R_0$ for the fundamental
initial data $\bx_0=(R_0,R_1)$. 
The quantum $A_1$ $Q$-system \cite{BZ,DFK10} expresses the mutations of the quantum cluster
algebra, and reads:
\begin{equation}\label{qaoneQsys}
q R_{j+1}R_{j-1}=R_j^2+1 \qquad (j\in \Z)
\end{equation}
Together with the above fundamental initial data, compatibility implies the following commutation relation holds
within each cluster $\bx_n=(R_n,R_{n+1})$:
\begin{equation}\label{comQ} R_n R_{n+1}=q \, R_{n+1}R_n
\end{equation}


Note that, like in the commuting case, the solutions of the quantum $A_1$ 
$Q$-system \eqref{qaoneQsys} are the solutions of the 
quantum $A_1$ $T$-system \eqref{qtsys} that are $2$-periodic in the index $i$, namely with
$R_j=T_{j\!\! \mod 2,j}$.

\subsubsection{Solution via quantum paths v/s quantum networks}
The quantum $A_1$ $Q$-system \eqref{qaoneQsys} 
was solved in \cite{DFK10} in terms of ``quantum" paths as follows. 
Let us define weights
\begin{equation}\label{ys} y_1=R_1R_0^{-1} \qquad y_2=R_1^{-1}R_0^{-1} \qquad y_3=R_1^{-1}R_0 \end{equation}
These weights satisfy the relations:
$$ y_1y_3=y_3y_1=q\qquad y_1y_2=q^2\,y_2y_1\qquad y_2y_3=q^2\,y_3y_2$$
We consider ``quantum" paths on the integer segment $[0,3]$ from and
to the origin $0$, with steps $\pm1$ with weight $1$ for $i\to i+1$, $i=0,1,2$ and
$y_i$ for $i\to i-1$, $i=1,2,3$. The weight of a path $p$, $w(p)$ is the product of the step weights
in the order in which they are taken. The partition function for quantum paths of length $2j$
is the sum over the paths $p$ from and to the origin, with $2j$ steps, of the weights $w(p)$.
We have the following 
\begin{thm}(\cite{DFK10})
For $j\in\Z_+$, the solution $R_j$ to the quantum $A_1$ $Q$-system \eqref{qaoneQsys} is equal to the partition 
function for quantum paths of length $2j$, times $R_0$.
\end{thm}

A reformulation of this result uses the ``two-step" transfer matrix $\cT$ whose entries are the weights of the 
paths of length $2$ that start (and end) at the even points $0$ and $2$:
\begin{equation}
\label{transmat} \cT=\begin{pmatrix} y_1 & 1 \\
y_2y_1 & y_2+y_3 \end{pmatrix} =\begin{pmatrix}
w\left(\raise -.2truecm \hbox{$ {\epsfxsize=1.cm \epsfbox{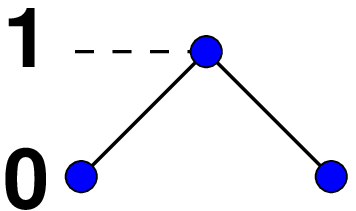}}$} \right) & 
w\left( \raise -.4truecm \hbox{$  {\epsfxsize=1.cm \epsfbox{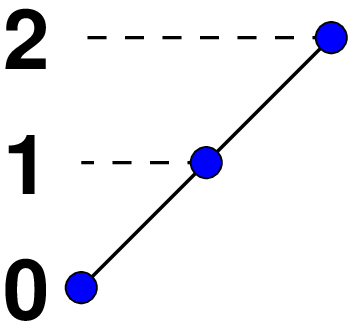}}$}\right) \\
w\left( \raise -.4truecm \hbox{$ {\epsfxsize=1.cm \epsfbox{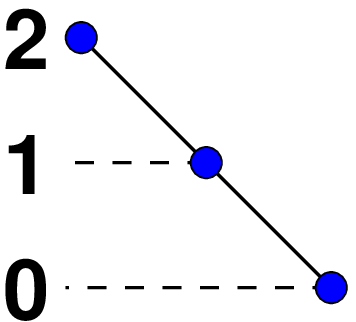}}$}\right) & 
w\left( \raise -.2truecm \hbox{$ {\epsfxsize=1.cm \epsfbox{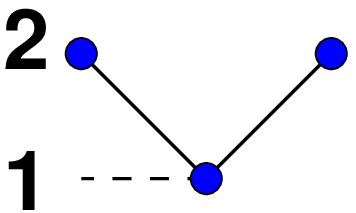}}$}\right) 
\end{pmatrix}
\end{equation}
The theorem may be rephrased as the following identity:
\begin{equation}
\label{rephrase}
 R_j =(\cT^j)_{1,1}\, R_0  \qquad (j\in \Z_+)
 \end{equation}

We may now apply the Theorem \ref{solq}
above to obtain an alternative quantum network formulation of the $A_1$ $Q$-system solutions. 
We start with the fundamental initial data $x_0=(R_0,R_1)$,
with $R_0R_1=qR_1R_0$. We have by Theorem \ref{solq}, for all $j\in\Z_{>0}$:
\begin{equation}\label{phrase}
R_j=T_{j\, {\rm mod}\, 2,j}=\big((VU)^{j-1}\big)_{1,1}R_1=\big((UV)^{j}\big)_{1,1}R_0
\end{equation}
where, due to the periodicity property, we have
\begin{eqnarray*}
&&V=V(R_1,R_0)=\begin{pmatrix}
R_1R_0^{-1} & R_0^{-1} \\
0 & 1 
\end{pmatrix}
\qquad U=U(R_0,R_1)=\begin{pmatrix}
1 & 0\\
q^{-1}R_1^{-1} & R_0R_1^{-1}
\end{pmatrix}\\
&& 
UV=\begin{pmatrix}
R_1R_0^{-1} & R_0^{-1} \\
q^{-1}R_0^{-1} & (R_0+R_0^{-1})R_1^{-1}
\end{pmatrix}=\begin{pmatrix}
y_1 & R_0^{-1} \\
q^{-1}R_0^{-1} & q^{-1}(y_2+y_3)
\end{pmatrix}
\end{eqnarray*}
Comparing with eq.\eqref{transmat}, and noting that 
$y_2y_1=R_1^{-1}R_0^{-1}R_1R_0^{-1}=q^{-1}R_0^{-2}$, and that
$R_0^{-1}(y_2+y_3)R_0=q^{-1}(y_2+y_3)$, we arrive at the relation:
\begin{equation}\label{ctud}
UV=\begin{pmatrix} 1 & 0\\ 0 & R_0^{-1}\end{pmatrix} \cT \begin{pmatrix} 1 & 0\\ 0 & R_0\end{pmatrix} 
\end{equation}
from which we deduce that Eqns. \eqref{rephrase} and \eqref{phrase} are equivalent, 
as the conjugation does not affect the $(1,1)$ matrix element.

\subsection{Non-commutative $Q$-system: a solution via non-commutative networks}
The fully non-commutative version of the $A_1$ $Q$-system was studied and solved in \cite{DFK09b}.
It reads:
\begin{equation}\label{konts} R_{n+1} R_n^{-1} R_{n-1}=R_n+R_n^{-1} \qquad (n\in\Z) 
\end{equation}
for $R_n$ some formal non-commutative variables subject to the quasi-commutation relations
\begin{equation}\label{ncqcom} R_n R_{n+1}=R_{n+1}C R_n \end{equation}
where $C$ is another fixed non-commuting variable. The quantum case is recovered when $C=q$ is central.

Using the relations \eqref{ncqcom},
the non-commutative $A_1$ $Q$-system \eqref{konts}
may be rewritten as
\begin{equation}\label{ncqsys} R_{n+1}CR_{n-1}=R_n^2+1 \qquad (n\in \Z) \end{equation}

Let us keep the definition \eqref{ys} for $y_1,y_2,y_3$, now
in terms of the non-commutative initial data $R_0$ and $R_1$.
The solution of \cite{DFK09b} involves ``non-commutative paths" from and to the origin on $[0,3]$, 
with weights $1$ for the steps $i\to i+1$, $i=0,1,2$
and $y_i$ for the steps $i\to i-1$, $i=1,2,3$. We have:
\begin{thm}(\cite{DFK09b})
For $n\in\Z_+$, the solution $R_n$ of the non-commutative $A_1$ $Q$-system
is the partition function for non-commutative paths from the origin
to itself with $2n$ steps, times $R_0$.
\end{thm}
As before, this is best expressed by use of the ``two-step" transfer matrix $\cT$,
still given by the expression \eqref{transmat} in terms of $R_0$ and $R_1$, and the
identity \eqref{rephrase} still holds.

The network solution of the quantum $A_1$ $Q$-system may be adapted for the 
fully non-commutative one as follows.
For non-commuting variables $a,b$, we introduce the matrices:
\begin{equation}\label{ncdu} 
U(a,b)=\begin{pmatrix} 1 & 0\\ C^{-1}b^{-1} & a b^{-1}\end{pmatrix} \qquad 
V(a,b)=\begin{pmatrix}a b^{-1} & b^{-1}\\ 0 & 1\end{pmatrix}
\end{equation}

We have the following generalization of Lemma \ref{prealem}.
\begin{lemma}\label{usefulem}
Assume $a,b,c$ have the quasi-commutations $ba=aCb$, $bc=cCb$, and $ac=ca$, then
we have the equation:
\begin{equation}\label{mutmatnc}
V(a,b)U(b,c)=U(a,b')V(b',c)
\end{equation}
for $b'$ defined via $b'Cb=ac+1$. 
Moreover, with this definition, $cb'=b'C c$ and $a b'=b'C a$.
\end{lemma}
\begin{proof}
We compute directly
$$V(a,b)U(b,c)=\begin{pmatrix}a b^{-1} & b^{-1}\\ 0 & 1\end{pmatrix}
\begin{pmatrix} 1 & 0\\ C^{-1}c^{-1} & b c^{-1}\end{pmatrix}
= \begin{pmatrix} (a+c^{-1})b^{-1} & c^{-1}\\ C^{-1}c^{-1} & b c^{-1}\end{pmatrix}$$
and
$$U(a,b')V(b',c)=\begin{pmatrix} 1 & 0\\ C^{-1}b'^{-1} & a b'^{-1}\end{pmatrix}
\begin{pmatrix}b' c^{-1} & c^{-1}\\ 0 & 1\end{pmatrix}
=\begin{pmatrix} b'c^{-1}& c^{-1}\\ C^{-1}c^{-1} & C^{-1}b'^{-1} c^{-1}+a b'^{-1}\end{pmatrix}$$
Identifying the two results, we obtain $b'c^{-1}bc=b'Cb=1+a c$ from the 
(1,1) element identification, and
$b'C b= 1+b'C a b'^{-1} c$ from the $(2,2)$ one. 
But from the first identity we deduce that $ab'=b'Ca$
(as well as $cb'=b'C c$), and the Lemma follows.
\end{proof}

\begin{thm}\label{NCmain}
For $n\in\Z_+$, $p\in\Z$, the solution $R_{n+p}$ of the non-commutative $A_1$ $Q$-system is expressed
in terms of the initial data $x_p=(R_p,R_{p+1})$ as:
$$ R_{n+p}=\big((V_pU_p)^{n-1}\big)_{1,1}R_{p+1}=\big((U_pV_p)^n\big)_{1,1}R_p$$
where 
$V_p=V(R_{p+1},R_p)$ and $U_p=U(R_p,R_{p+1})$ in terms of the matrices of eq.\eqref{ncdu}.
\end{thm}
\begin{proof}
By induction on $n$. The result is clear for $n=1$.
Assume that the Theorem holds for $n-1,p+1$.
Applying the Lemma \ref{usefulem} for $a=c=R_{p+1}$ and $b=R_{p}$, we get $x=R_{p+2}$ and
$V_{p}U_{p}=U_{p+1}V_{p+1}$. We deduce that:
$$\big((V_pU_p)^{n-1}\big)_{1,1}R_{p+1}=\big((U_{p+1}V_{p+1})^{n-1}\big)_{1,1}R_{p+1}=R_{(n-1)+(p+1)}=R_{n+p}$$
and the Theorem holds for $n,p$. Finally, the translational invariance of \eqref{ncqsys}
implies that $R_{n+p}$ is the same {\it function} of $(R_p,R_{p+1})$ as $R_n$ of $(R_0,R_1)$, independently of $p\in\Z$.
The Theorem follows.
\end{proof}

The direct connection between the non-commutative network formulation and the non-commutative
path formulation is the same relation \eqref{ctud} between $VU$ and $\cT$ as in the quantum  case.
It is now a consequence of: $y_2y_1=R_1^{-1}R_0^{-1}R_1R_0^{-1}=R_0^{-1}C^{-1}R_0^{-1}$
and $R_0(y_2+y_3)R_0^{-1}=C^{-1}R_1^{-1}R_0^{-1}+R_0R_1^{-1}$, by use of the quasi-commutations
\eqref{ncqcom}.

%
%

\section{Discussion: connection to the quantum lattice Liouville equation}\label{conclusion}



In this paper, we have introduced and solved the quantum $A_1$ $T$-system
in terms of an arbitrary admissible data, by means of a quantum path model. 
This system turns out to be closely related 
to the  quantum lattice Liouville equation of \cite{FKV}. 

The $T$-systems in general are related to the so-called $Y$-systems via a
birational transformation \cite{KNS}. The following is a quantum
version of this transformation.  Define
$$\chi_{i,j}=\left( T_{i+1,j}T_{i-1,j}\right)^{-1}\qquad (i,j\in \Z).$$
Then
\begin{eqnarray*}
\chi_{i,j-1}\chi_{i,j+1}&=&T_{i+1,j-1}^{-1}(T_{i-1,j+1}T_{i-1,j-1})^{-1} T_{i+1,j+1}^{-1}\\
&=&qT_{i+1,j-1}^{-1} (1+T_{i,j}T_{i-2,j})^{-1}T_{i+1,j+1}^{-1}\\
&=&q (T_{i+1,j+1}T_{i+1,j-1})^{-1}(1+T_{i,j}T_{i-2,j})^{-1}\\
&=&q^2 \chi_{i+1,j}(1+\chi_{i+1,j})^{-1}\chi_{i-1,j}(1+\chi_{i-1,j})^{-1}.
\end{eqnarray*}
Note that all the factors in the last line commute. 

This system may be called quantum $Y$-system for $A_1$:
If we set $q=1$ and denote $Y_i(j/2)=\chi_{i,j}$ it coincides with the
$A_1$ $Y$-system \cite{Zam,RTV}.  In the non-commuting case, it coincides
(upon a reordering of the left hand side and a renormalization of the
variables, $\chi\mapsto q \chi$) with the quantum lattice Liouville
equations of \cite{FKV}.

The variables $\chi_{i,j}$ have local commutation relations within
each cluster. 
\begin{lemma}
$$ \chi_{k\pm 1,l+1}\chi_{k,l}=q^2 \chi_{k,l}\chi_{k\pm 1,l+1}, $$
while all other pairs of variables commute.
\end{lemma}
\begin{proof}
Let $j\geq l$. If $(i,j)\neq (k,l+1)$, then the commutation relations
\eqref{qcomm} imply that
$T_{i,j}\chi_{k,l}=\chi_{k,l}T_{i,j}.$
Otherwise,
$$ T_{k,l+1} \chi_{k,l}=T_{k,l+1}T_{k+1,l}^{-1}T_{k-1,l}^{-1}=q^2\chi_{k,l}T_{k,l+1},$$
and the Lemma follows.
\end{proof}

Note that here we do not impose the periodicity condition of
\cite{FKV}, that $\chi_{i+2N,j}=\chi_{i,j}$.  Here it would be
implemented by a periodicity condition of the form
$T_{i+2N,j}=T_{i,j}$.  The solution of this paper holds for any choice of
boundary conditions, and may therefore be {\it restricted} to these
special periodicity conditions on $T_{i,j}$.

%
%
%
%
%



\begin{thebibliography}{10}

\bibitem{FRISES}I. Assem, C. Reutenauer, and D. Smith {\em Frises}, 
{\tt arXiv:0906.2026 [math.RA]}.




\bibitem{BZ}A. Berenstein, A. Zelevinsky, \emph{Quantum Cluster Algebras},
Adv. Math. {\bf 195} (2005) 405--455. 
{\tt arXiv:math/0404446 [math.QA]}.

  
\bibitem{POSIT} P. Caldero and M. Reineke, \emph{On the quiver
    Grassmannian in the acyclic case}.  J. Pure Appl. Algebra   \textbf{212}
  (2008),  no. 11, 2369--2380. {\tt arXiv:math/0611074
   [math.RT]}.

    
\bibitem{Cox} H.S.M. Coxeter, {\em Frieze Patterns, Triangulated Polygons and Dichromatic Symmetry}, 
in {\it The Lighter Side of Mathematics}, R.K. Guy and E. Woodrow (eds.), John Wiley $\&$ Sons, NY, (1961) pp 15-27.

\bibitem{DF} P. Di Francesco, {\em The solution of the $A_r$ T-system for arbitrary boundary},
Elec. Jour. of Comb. Vol. {\bf 17(1)} (2010) R89. 
{\tt arXiv:1002.4427 [math.CO]}.

\bibitem{DFK08} P. Di Francesco and R. Kedem, \emph{Q-systems as
cluster algebras II},
Lett. Math. Phys. {\bf 89} No 3 (2009) 183-216. 
{\tt arXiv:0803.0362 [math.RT]}.
    
\bibitem{DFK3} P. Di Francesco and R. Kedem, \emph{Q-systems, heaps,
paths and cluster positivity}, Comm. Math. Phys. {\bf 293} No. 3 (2009) 727--802,
DOI 10.1007/s00220-009-0947-5.
{\tt arXiv:0811.3027 [math.CO]}.
%
%

\bibitem{DKTsys} P. Di Francesco and R. Kedem, \emph{Positivity of the
$T$-system cluster algebra},  Elec. Jour. of Comb. Vol. {\bf 16(1)}
(2009) R140, Oberwolfach preprint OWP 2009-21, 
{\tt arXiv:0908.3122 [math.CO]}.
%
\bibitem{DFK09b}P. Di Francesco and R. Kedem, 
\emph{Discrete non-commutative integrability: proof of a 
conjecture by M. Kontsevich}, Int. Math. Res. Notices (2010),  doi:10.1093/imrn/rnq024.
{\tt arXiv:0909.0615 [math-ph]}.

\bibitem{DFK10}P. Di Francesco and R. Kedem, {\em Noncommutative integrability, 
paths and quasi-determinants}, preprint
{\tt arXiv:1006.4774 [math-ph]}.

\bibitem{FV} L.D. Faddeev and A.Y. Volkov {\em
Discrete evolution for the zero modes of the quantum Liouville model.}
J. Phys. A {\bf 41} (2008), no. 19, 194008, 12 pp. 
{\tt arXiv:0803.0230 [hep-th]}.

\bibitem{FKV} L. Faddeev, R. Kashaev and V. Volkov, {\em Strongly coupled
quantum discrete Liouville theory. I: Algebraic approach and duality},
Comm. Math. Phys. {\bf 219} No 1 (2001) 199-219.
{\tt arXiv:hep-th/0006156}.


\bibitem{FZI} S. Fomin and A. Zelevinsky Cluster Algebras I.
  J. Amer. Math. Soc.   \textbf{15}  (2002),  no. 2, 497--529 {\tt
    arXiv:math/0104151 [math.RT]}.

\bibitem{FZLaurent} S. Fomin And A. Zelevinsky \emph{The Laurent
    phenomenon}.   Adv. in Appl. Math.   \textbf{28}  (2002),  no. 2,
  119--144. {\tt arXiv:math/0104241 [math.CO]}. 
  
\bibitem{FZIV} S. Fomin and A. Zelevinsky Cluster Algebras IV: coefficients.
Compos. Math. \textbf{143} (2007), 112--164.  {\tt
    arXiv:math/0602259 [math.RA]}.
  
\bibitem{FR} E. Frenkel and N. Reshetikhin, {\em The
  $q$-characters of representations of quantum affine algebras and
  deformations of $W$-algebras}. In { Recent developments in quantum affine
  algebras and related topics (Raleigh, NC 1998)}, Contemp. Math. {\bf
  248} (1999),
  163--205.





	


	


\bibitem{Ke07} R. Kedem, \emph{$Q$-systems as cluster algebras}.  J.
  Phys. A: Math. Theor. \textbf{41} (2008) 194011 (14 pages). {\tt
    arXiv:0712.2695 [math.RT]}.

\bibitem{KR} A.~N. Kirillov and N.~Yu. Reshetikhin,
  \emph{Representations of Yangians 
and multiplicity of occurrence of the irreducible components of the
tensor product 
of representations of simple {L}ie algebras}, J. Sov. Math. {\bf 52}
(1990) 3156-3164. 


\bibitem{KT}A. Knutson, T. Tao, and C. Woodward, 
\emph{A positive proof of the Littlewood-Richardson rule using the
  octahedron recurrence},   
Electr. J. Combin. \textbf{11} (2004) RP 61. 
{\tt arXiv:math/0306274 [math.CO]}
	    


\bibitem{KLWZ} I. Krichever, O. Lipan, P. Wiegmann and A. Zabrodin, 
{\em Quantum Integrable Systems and Elliptic Solutions of Classical 
Discrete Nonlinear Equations},
Comm. Math. Phys. {\bf 188} (1997), 267--304.
{\tt arXiv:hep-th/9604080}. 

\bibitem{KNS} A. Kuniba, A. Nakanishi and J. Suzuki, {\em Functional relations 
in solvable lattice models. I. Functional relations and representation theory.} 
International J. Modern Phys. A {\bf 9} no. 30, pp 5215--5266 (1994).
{\tt arXiv:hep-th/9310060}. 



\bibitem{MSW} G. Musiker, R. Schiffler and L. Williams, 
{\em Positivity for cluster algebras from surfaces},
preprint
{\tt arXiv:0906.0748 [math.CO]}.

	
  
	

%



\bibitem{SPY} D. Speyer, \emph{Perfect matchings and the
        octahedron recurrence}, J. Algebraic Comb. \textbf{25} No 3
      (2007) 309-348. {\tt arXiv:math/0402452 [math.CO]}.



\bibitem{RTV} F. Ravanini, R. Tateo, and A. Valleriani,
    {A new family of diagonal {$ADE$}-related scattering theories},
   {\em Phys. Lett.}{\bf B 293} (1992), {361--366}.

\bibitem{Zam} Al. B. Zamolodchikov, 
     {On the thermodynamic {B}ethe ansatz equations for
              reflectionless {$ADE$} scattering theories},
{\em Physics Letters} {\bf B 253}, {1991},
{391--394}.

\end{thebibliography}
\end{document}